\documentclass[11pt, letterpaper]{article}
\usepackage[margin=1.2in,footskip=0.25in]{geometry}

\usepackage{amsmath, amssymb, mathtools}
\usepackage{amsthm}
\usepackage{bm}
\allowdisplaybreaks
\usepackage{xpatch}
\makeatletter
\xpatchcmd{\@thm}{\thm@headpunct{.}}{\thm@headpunct{}}{}{}
\makeatother

\usepackage[labelfont=bf]{caption}
\usepackage{subcaption}
\usepackage[table]{xcolor}

\usepackage[T1]{fontenc}
\usepackage[scaled=0.8]{beramono}

\usepackage{xcolor}
\definecolor{linkcolor}{HTML}{0645AD}
\definecolor{crimson}{HTML}{A41034}
\definecolor{indianred}{HTML}{CD5C5C}
\definecolor{dodgerblue}{HTML}{1E90FF}

\usepackage{hyperref}
\hypersetup{
     colorlinks   = true,
     allcolors    = crimson,
}
\usepackage{cleveref}

\def\*#1{\mathbf{#1}}

\newcommand{\E}{\mathbb{E}}

\newcommand{\indep}{\!\perp\!\!\!\perp}
\newcommand{\lb}{\mathsf{LB}}
\newcommand{\ub}{\mathsf{UB}}

\newtheorem{theorem}{Theorem}
\newtheorem{proposition}{Proposition}
\newtheorem{corollary}{Corollary}
\newtheorem{lemma}{Lemma}
\newtheorem{assumption}{Assumption}
\newtheorem{remark}{Remark}
\newtheorem{condition}{Condition}
\newtheorem{result}{Result}
\newtheorem{example}{Example}
\theoremstyle{definition}
\newtheorem{definition}{Definition}[section]

\makeatletter
\newenvironment{namedassumption}[3]{
  \global\expandafter\def\csname title@#1\endcsname{#2}%
  \begin{assumption}\textnormal{(}#2\textnormal{)}\label{#1}. \itshape #3\end{assumption}%
}{}

\newenvironment{namedtheorem}[3]{%
  \global\expandafter\def\csname title@#1\endcsname{#2}%
  \begin{theorem}\textnormal{(}#2\textnormal{)}\label{#1}. \itshape #3\end{theorem}%
}{}

\newenvironment{namedproposition}[3]{%
  \global\expandafter\def\csname title@#1\endcsname{#2}%
  \begin{proposition}\textnormal{(}#2\textnormal{)}\label{#1}. \itshape #3\end{proposition}%
}{}

\newenvironment{namedremark}[3]{%
  \global\expandafter\def\csname title@#1\endcsname{#2}%
  \begin{remark}\textnormal{(}#2\textnormal{)}\label{#1}. \itshape #3\end{remark}%
}{}

\newcommand{\namedcref}[1]{%
  \cref{#1} (\@nameuse{title@#1})}
\newcommand{\namedCref}[1]{%
  \Cref{#1} (\@nameuse{title@#1})}
\makeatother

\numberwithin{equation}{section}

\usepackage{setspace}
\usepackage{pdfpages}
\usepackage{booktabs}
\usepackage{makecell}
\usepackage{float}

\usepackage{tikz}
\usetikzlibrary{positioning}

\usepackage{natbib}

\title{Difference-in-differences Design with Outcomes\\Missing Not at Random}
\author{Sooahn Shin\thanks{Ph.D. Candidate, Department of Government and Institute for Quantitative Social Science, 
Harvard University.  Email: \url{sooahnshin@g.harvard.edu}. URL: \url{https://sooahnshin.com}.}
\thanks{I am grateful to Kosuke Imai, Matthew Blackwell, Naijia Liu, Soichiro Yamauchi, Imai Research Group, Andrew Q. Philips, Anton Strezhnev, and the participants at the 81st MPSA Annual Conference and 120th APSA Annual Meeting and Exhibition for their helpful comments and suggestions. All errors are my own.}
  }
\date{
    \centering
    \begin{tabular}{@{}l@{}}
        First version: April 1, 2024\\
        This version: September 16, 2024
    \end{tabular}
}

\begin{document}
\maketitle
\begin{abstract}
    This paper addresses one of the most prevalent problems encountered by political scientists working with difference-in-differences (DID) design: missingness in panel data.
    A common practice for handling missing data, known as complete case analysis, is to drop cases with any missing values over time.
    A more principled approach involves using nonparametric bounds on causal effects or applying inverse probability weighting based on baseline covariates.
    Yet, these methods are general remedies that often under-utilize the assumptions already imposed on panel structure for causal identification.
    In this paper, I outline the pitfalls of complete case analysis and propose an alternative identification strategy based on principal strata.
    To be specific, I impose parallel trends assumption within each latent group that shares the same missingness pattern (e.g., always-respondents, if-treated-respondents) and leverage missingness rates over time to estimate the proportions of these groups.
    Building on this, I tailor Lee bounds, a well-known nonparametric bounds under selection bias, to partially identify the causal effect within the DID design.
    Unlike complete case analysis, the proposed method does not require independence between treatment selection and missingness patterns, nor does it assume homogeneous effects across these patterns. 
\end{abstract}


\setstretch{1.5}

\newpage
\section{Introduction}
Difference-in-differences (DID) design is a quasi-experimental method in social science widely used to estimate causal effects of a treatment on an outcome variable using repeated observations of units over time. The DID design is particularly useful when the treatment is not randomly assigned, and the researcher is concerned with unobserved time-invariant confounder. Yet, one of the most prevalent problems encountered by researchers working with panel data is missingness. For example, this problem is evident in DID studies that utilize survey data to measure outcome variables, where respondent's non-response is a frequent concern. \cite{chiu2023}, in their extensive review and replication of articles from three leading political science journals using observational panel data with binary treatments, highlighted this issue with unbalanced panels. They especially emphasized that the missingness pattern is appeared to be nonrandom or extremely prevalent in some studies.

Despite the prevalence of missing data in panel studies, most methodological work presumes balanced panels without missing data. 
When faced with the methodological challenge of addressing potential bias due to missing data, researchers have often resorted to complete case analysis (listwise deletion) or imputation methods. 
However, these methods are not without their own limitations. 
Complete case analysis can lead to biased estimates when the missingness is not completely at random (MCAR) or missing at random (MAR). 
Imputation methods can also be problematic when the missingness is not at random (MNAR) or the covariates are also susceptible to missingness. 
Alternatively, when the quantity of interest is a causal estimand, a more principled approach to address missing data is to use nonparametric bounds of causal effects \citep{horowitz2000, zhang2003, imai2008, lee2009} or inverse probability weighting using baseline covariates. However, these methods are general remedies that under-utilize the assumptions already imposed on panel structure for causal identification.

This implies that the intersection of causal inference, panel data, and missing values presents a unique challenge in social science studies that remains underexplored. 
Recently, there have been attempts in disciplines adjacent to social science to address attrition bias by using parallel trends assumptions and the changes-in-changes approach \citep{ghanem2022, dukes2022}.
\cite{ghanem2022} extends the changes-in-changes condition so that the distribution of unobserved heterogeniety to be stable across time within treatment-response subpopulations. Using this assumption, they identify the average treatment effect of the treated-respondents and also show that the average treatment effect can be identified with an additional assumption that the distribution is homogeneous across treatment-response subpopulations. \cite{dukes2022} consider two alternative strategies, one based on the parallel trends assumption and the other based on the `bespoke instrumental variable' approach, yet their approach is limited to randomized experiments. 

In this paper, I provide an alternative approach of addressing missing outcome variable in panel data with DID design for observational studies. My discussion aligns with the recent studies but also differs from them by setting the DID design and parallel trends assumption for \textit{causal identification} at the center and exploring the assumptions and data structure within the DID design to address the \textit{missing data problem}. Specifically, I answer the following questions using the DID design and the principal stratification framework: Under which extension of parallel trends assumptions can we justify the complete case analysis? What are the main identification challenges with missing data in the DID design? How can we identify the average treatment effect for treated (ATT) using auxiliary variables from the panel data that offer additional information yet have not been explored?

What must not be overlooked is that missing indicator in this setup is a post-treatment intermediate variable. In this vein, I proceed to introduce principal stratification \citep{frangakis2002}--namely, always-respondents, if-treated-respondents, if-control-respondents, and never-respondents--to induce parallel trends assumptions conditioning on these latent groups that shares the same missingness pattern. I interpret the complete case analysis under this framework and discuss the main identification challenges with missing data: (1) potential dependence between selection into treatment with principal strata and (2) heterogenous effect across principal strata. I then extend the DID design using auxiliary variables (e.g. outcome variables in multiple pretreatment periods), and use its response indicator as an instrumental variable for identifying ATT. The essential intuition behind this approach is that the response indicator of the auxiliary variables can be used as an instrumental variable that affects the time trend of the outcome variable only through the treatment selection and the missingness of the post-treatment outcome variable. Lastly, I propose alternative approaches with the principal strata specific parallel trends assumption to partially identify the pricipal strata specific ATT. 
To be specific, I impose parallel trends assumption within each principal stratum and leverage missingness rates over time to estimate the proportions of these groups.
Building on this, I tailor Lee bounds \citep{lee2009}, a well-known nonparametric bounds under selection bias, to partially identify the causal effect within the DID design.

\section{Problems of the Standard Approach} \label{sec:standard}

In this section, I review the parallel trends assumptions under which the standard DID design with complete case analysis can be justified. Using the principal stratification framework, I describe the limitations of this standard methodology and discuss the main identification challenges with missing data in the DID design.

\subsection{Motivating Applications}

To illustrate the problem of missing data in the DID design, I revisit two application studies using two-periods DID design with survey data. In the first application, I revisit \cite{sexton2023} which studies how small development aid projects affect the public perception and attitudes toward government. Specifically, the authors are interested in the impacts of German development aid during 2017–18 on subsequent political attitudes in northern Afghanistan. The study uses two waves (2016 and 2018) of geocoded public opinion survey data with five different main outcome variables. As shown in Table \ref{tab:sexton_missing}, the missingness of the survey responses in the second wave is substantial, and the pattern of missingness appears to be different across outcome variables. Particularly, the missingness ratio differ in time trend (\textit{before-after}), difference between treated and control groups (\textit{treated-control}), and its interaction (\textit{differece-in-differences}). 

\begin{table}[ht!]
    \centering
    \begin{tabular}{ccccccc}
\toprule
Wave & \makecell{Treatment\\Group} & \makecell{Afghanistan\\Right\\Direction?} & \makecell{Confidence\\in\\President} & \makecell{Local\\Government\\Confidence} & \makecell{National\\Governemnt\\Good Job} & \makecell{Sympathy\\for\\Insurgents}\\
\midrule
2016 & Control & 4.52 & 0.37 & 18.57 & 0.47 & 4.63\\
2016 & Treated & 4.74 & 0.55 & 69.83 & 0.14 & 7.01\\
2018 & Control & 8.14 & 1.83 & 14.60 & 0.22 & 3.12\\
2018 & Treated & 8.28 & 0.96 & 64.83 & 1.14 & 2.34\\
\bottomrule
\end{tabular}
    \caption{Missingness of Survey Responses in \cite{sexton2023} in Percentage Points}
    \label{tab:sexton_missing}
\end{table}

In the second study, I revisit \cite{bisgaard2018} which examines the effects of elite partisan cues on economic perception. The study uses five waves of panel surveys collected from 2010 to 2011 in Denmark to track public opinion on economic issues. After the second wave of survey data, the Center-Right government in Denmark dramatically changed its partisan cue on the severity of the public budget deficit, which led to a change in the economic perception of incumbent supporters according the their findings. The missingness of the survey responses in this study is shown in Table \ref{tab:bisgaard_missing}. Given a non-ignorable amount of missingness, the authors provided additional regression analysis with the second and third waves, and concluded the missingness does \textit{not} appear to be systematically different across treatment groups and time periods, conditioning on prior perceptions of the national economy.

\begin{table}[ht!]
    \centering
    \begin{tabular}{rlrr}
    \toprule
    Wave & Period & \makecell{Treated Group\\(Incumbent Supporters)} & \makecell{Control Group\\(Opposition Supporters)}\\
    \midrule
    1 & Pre-treatment & 12.82 & 16.42\\
    2 & Pre-treatment & 39.16 & 42.26\\
    3 & Post-treatment & 44.87 & 45.72\\
    4 & Post-treatment & 47.90 & 48.40\\
    5 & Post-treatment & 54.31 & 54.45\\
    \bottomrule
    \end{tabular}
    \caption{Missingness of Survey Responses in \cite{bisgaard2018} in Percentage Points}
    \label{tab:bisgaard_missing}
\end{table}

Several questions arise from these examples, particularly related to the unique features of panel data. 
What does the trend of missingness within each treatment group imply for potential bias in the DID estimates? 
Can we directly compare such trends between different treatment groups, and if they are similar across groups, would the DID estimate from complete case analysis be unbiased? 
Answering these questions necessitates exploring diverse variants of canonical parallel trends assumptions and their substantive implications. 
As I will discuss in the following sections, the short answer is no. 
It requires an additional assumption regarding the parallel trends of the outcome variable between respondents and nonrespondents within each treatment group, which may restrict the heterogeneity of the treatment effect.

Another important aspect of these questions is how we define the groups in terms of missingness and how the composition of these groups relates to causal identification. 
In particular, it is crucial to recognize that the missingness of the outcome variable is a post-treatment variable, meaning there exists a latent group of units with different missingness patterns under different treatment statuses. 
This motivates the need for principal stratification, which examines the potential missingness patterns of the outcome variable under different treatment statuses. 
Given that we are interested in observational studies where the treatment is not randomly assigned, the composition of these latent groups may vary across treatment groups. 
This may be due to selection bias, where the treatment selection is correlated with missingness, heterogenous treatment effect across these latent groups, or both.

Lastly, it is worth noting that panel data provides additional information that can be crucial for identifying causal effects. At a minimum, researchers always have access to the missingness rate of the outcome variable over time, which can be used in causal identification. In more favorable cases, researchers may also have access to auxiliary variables from the panel data, such as other outcome measures or multiwave data, which can further aid in identifying the causal effect. In this paper, I propose a novel identification strategy that combines the principal strata framework with these auxiliary variables to address the missing data problem in the DID design.

\subsection{Setup and Notation}
I consider a two-period DID design with binary treatment and missingness in the outcome variable. Let $D_{i}$ be a binary treatment group indicator of unit $i$, where $D_{i} = 1$ for the treated group and $0$ for the control group. Let $Y_{it}$ be the outcome variable of unit $i$ at time $t = 1,2$, where time $1$ is the pre-treatment period and $2$ is the post-treatment period. I assume the consistency and no anticipation assumptions for the potential outcomes:
\begin{align*}
    &Y_{it} = D_{i} Y_{it}(1) + (1-D_{i}) Y_{it}(0) \text{ for } t = 1,2 &\text{ (Consistency)}\\
    &Y_{it}(1) = Y_{it}(0) \text{ for } t = 1 &\text{ (No anticipation of treatment)}
\end{align*}
Let $R_{i}$ denote a binary response indicator at time $t = 2$. $R_{i} = 1$ if $Y_{i2}$ is observed, and $0$ if it is missing. Here, $R_{i}(d)$ is the potential response indicator if $D_{i} = d$. I assume the same consistency and no anticipation assumptions for $R_{i}$.

One example of this setting is a study using a DID design (or two-way fixed effects) where the data comes from a survey conducted in two waves, with attrition observed in the second wave. Another potential example of missing outcomes is the ``don't know'' response to survey questions. A common practice in this case is to treat ``don't know'' as missing, resulting in a similar setup to survey data with attrition. Although I illustrate the problem of missing data in the DID design using survey attrition to facilitate understanding, the proposed methodology can be applied to other types of missing data and is not limited to surveys.

Based on the joint distribution of the response indicator and treatment group indicator we can define the following groups:
\begin{align*}
    (R_{i}, D_{i}) = 
    \begin{cases}
        (1,1) & \text{``treated-respondents''} \\
        (0,1) & \text{``treated-nonrespondents''} \\
        (1,0) & \text{``control-respondents''} \\
        (0,0) & \text{``control-nonrespondents''}
    \end{cases}
\end{align*}
One critical limitation of such a group is that it does not account for the fact that the response itself is a post-treatment variable. This group is considered crude because it only captures the realized response under a given treatment assignment and does not consider the counterfactual response (i.e., whether these individuals would have responded if they had been selected into the other treatment group). For example, ``treated-respondents'' in the first motivating example correspond to those people whose districts received aid and responded to the survey question afterward. We do not know whether this group of people would have also responded to the question if their district had not received such aid.

Alternatively, we can introduce the principal strata framework \citep{frangakis2002} in this setup based on the joint distribution of the potential outcomes of the response indicator. 
Let $S_i = (R_{i}(1), R_{i}(0))$ denote a principal strata define as below. For example, $S_i = (1,1)$ implies that unit $i$'s outcome is always observed, no matter the treatment status.
\begin{equation*}
    S_{i} \equiv (R_{i}(1), R_{i}(0)) = 
    \begin{cases}
        (0,0) & \text{``never-respondents''} \\
        (1,0) & \text{``if-treated-respondents''} \\
        (0,1) & \text{``if-control-respondents''} \\
        (1,1) & \text{``always-respondents''}
    \end{cases}
\end{equation*}
Note that this can be viewed as a latent pre-treatment covariate, distinct from the observed post-treatment response indicator $R_{i}$. 
With this, we can consider the joint distribution of this principal stratum and treatment group (e.g. never-respondents who are treated), as opposed to the realized response and treatment group  (e.g. treated-respondents, which is a mix of never-respondents and if-treated-respondents within the treated group). We use the term ``respondents'' for ease of interpretation, but it may not necessarily refer to survey respondents. In general, this term should be understood as a group of units that exhibit different missingness patterns under two possible treatment regimes.

The principal strata framework is widely utilized in the causal inference literature, particularly for examining issues of truncation by death and noncompliance. In the context of truncation by death in medical studies, for example, ``always-respondents'' corresponds to ``survivors,'' who would have always survived no matter what the medical treatment (e.g. an uptake of a medicine) was. In the context of noncompliance, ``always-respondents'' correspond to ``always-taker,'' who would have always taken the medicine no matter what the treatment assignment, an encouragement to take the medicine, was.

Similar to these issues, the potential outcome of the response indicator in our context is significant because it defines latent subgroups of units with substantively different characteristics. For example, in the first motivating application, if-treated-respondents represent a group of people who respond to these sensitive survey questions only if their district received aid. In contrast, if-control-respondents are those who respond to the survey questions only if their district did not receive the aid. One can imagine that there might be a systematic reason for these opposing behaviors, which could be related to the impact of aid. In the second motivating application, always-respondents are individuals who consistently respond to survey questions about economic issues, regardless of partisan cues. In contrast, never-respondents are those who never respond to the survey questions on economic issues. These two groups may possess distinct characteristics, such as varying levels of political engagement or interest in economic matters, which could correlate with partisanship. 

Two main issues with different principal strata are that (1) the proportions of these latent groups may vary between the treated and control groups, and (2) the treatment effect may differ across these latent groups. I will formally demonstrate this intuition in this section, where I discuss the identification challenges within this principal strata framework. To be specific, in this paper, I treat the principal strata as one of the conditioning variables when imposing parallel trends assumptions for causal identification, a topic that will be elaborated upon in this section. For simplicity, I assume there is no missing data in the first wave (alternatively, one could assume MCAR for the pre-treatment outcome missingness for now). However, this assumption will be relaxed in the following section with the proposed solution to the missing data problem.

\subsection{Complete Case Analysis and Parallel Trends Assumptions}

To begin with, we discuss how the general practice of complete case analysis can be justified under the parallel trends assumptions. 
In the DID design, researchers are interested in identifying the ATT defined as
\begin{equation*}
    \text{ATT} = \E[Y_{i2}(1) - Y_{i2}(0) \mid D_{i} = 1]
\end{equation*}
We employ the canonical parallel trends assumption to identify the counterfactual outcome $Y_{i2}(0)$ for treated units.
\begin{namedassumption}{a:pt-outcome}{Parallel trends}{
    \begin{equation*}
        \E[Y_{i2}(0) - Y_{i1} \mid D_{i} = 1] = \E[Y_{i2}(0) - Y_{i1}  \mid D_{i} = 0]
    \end{equation*}}
\end{namedassumption}
With this assumption and the consistency assumption, we have
\begin{align}
    \text{ATT} &= \E[Y_{i2} - Y_{i2}(0) \mid D_{i} = 1] \nonumber \\
    &=\E[Y_{i2} - Y_{i1} \mid D_{i} = 1] - \E[Y_{i2}(0) - Y_{i1} \mid D_{i} = 1] \nonumber \\
    &= \E[Y_{i2} - Y_{i1} \mid D_{i} = 1] - \E[Y_{i2} - Y_{i1} \mid D_{i} = 0] \label{eq:ident-att}
\end{align}
The final expression is not identified with missing data in the outcome variable. Upon the missingness in panel data with a randomized treatment, \cite{dukes2022} consider imposing the parallel trends assumption between the respondents and nonrespondents, for the time trend under each treatment status respectively.
\begin{namedassumption}{a:pt-observed}{Parallel trends of observed outcome between respondents and nonrespondents}{
    \begin{equation*}
        \E[Y_{i2}(d) - Y_{i1} \mid R_{i} = 1, D_{i} = d] = \E[Y_{i2}(d) - Y_{i1}  \mid R_{i} = 0, D_{i} = d]
    \end{equation*}
    for $d \in \{0,1\}$.}
\end{namedassumption}
What this assumption implies is as follows. We first consider the expression under the control group, $d = 0$. 
\begin{equation*}
    \E[\underbrace{Y_{i2}(0) - Y_{i1}}_{\text{time trend}} \mid R_{i} = 1, D_{i} = 0] = \E[Y_{i2}(0) - Y_{i1}  \mid R_{i} = 0, D_{i} = 0]
\end{equation*}
It assumes that the time trend of outcome, $Y_{i2}(0) - Y_{i1}$, is on average parallel between control respondents and control nonrespondents. For example, in the motivating application, within the districts that received aid, the average change in perception over time among respondents is equal to that among nonrespondents. This can be considered a variant of the canonical parallel trends assumption, where we are interested in DID of control-respondent and control-nonrespondents groups. Since the assumption is made on the time trends, not the outcome levels, it allows for the missingness to be correlated with the outcome level. For instance, if respondents who are more positive towards the government are more likely to respond to the survey, yet the change in perception over time remains constant, the assumption may still hold. Additionally, if there are multiple pre-treatment periods without missingness, this assumption can be indirectly tested by examining the parallel trends of the outcome between control-respondent and control-nonrespondents groups.

Next, let's consider \Cref{a:pt-observed} under the treated group, $d = 1$.
\begin{align*}
    &\E[\underbrace{(Y_{i2}(1) - Y_{i2}(0))}_{\text{causal effect}} + \underbrace{(Y_{i2}(0) - Y_{i1})}_{\text{time trend}}  \mid R_{i} = 1, D_{i} = 1] \\
    &= \E[(Y_{i2}(1) - Y_{i2}(0)) + (Y_{i2}(0) - Y_{i1})  \mid R_{i} = 0, D_{i} = 1]
\end{align*}
This assumption is more restrictive than the former, as it requires the parallel trends of the \textit{sum} of causal effect and time trend of the outcome between treated-respondents and treated-nonrespondents. For example, within the districts that received the aid, the change in perception pertaining to the aid and a common time shock among the respondents is on average equivalent to that of the nonrespondents. Hypothetically, this can be violated if the missingness is correlated with the treatment effect size, holding the time trend constant. For example, if the respondents who received the aid and became more positive towards the government are more likely to respond to the survey, the assumption may be violated.

\begin{figure}[ht!]
    \centering
    \includegraphics[width = \linewidth]{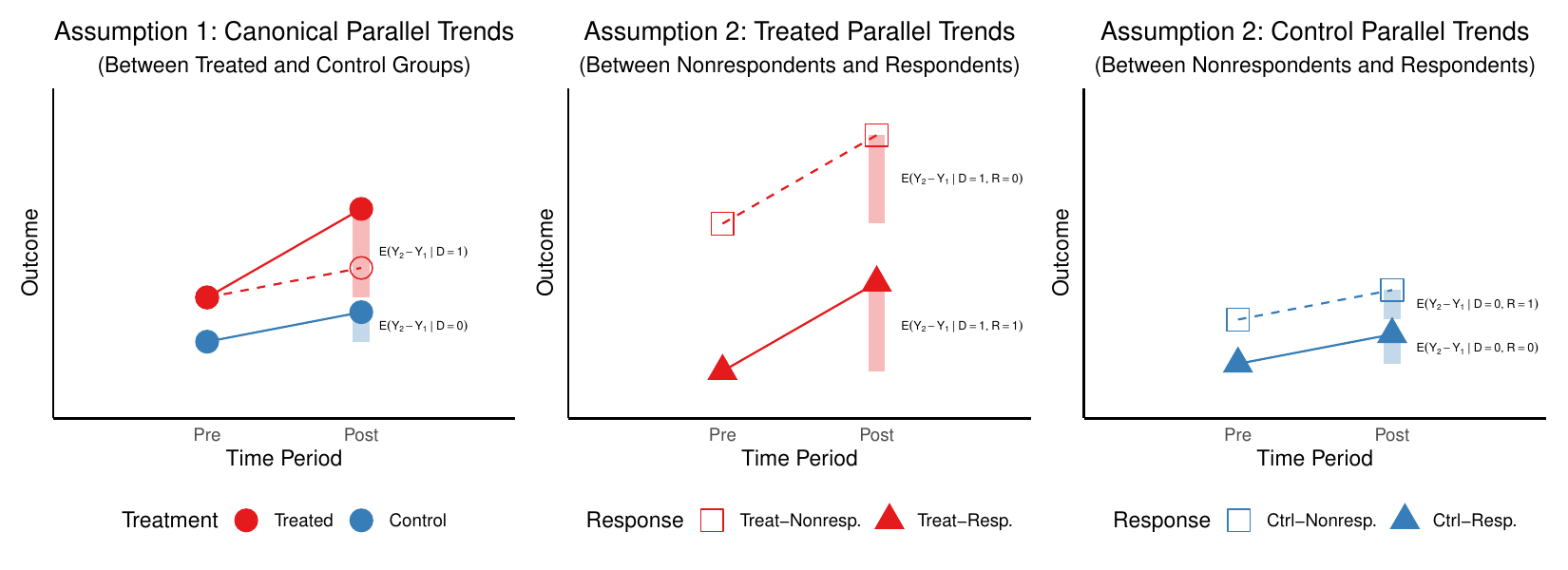}
    \caption{Illustration of the Parallel Trends \Cref{a:pt-outcome} and \Cref{a:pt-observed}. (Left): Canonical parallel trends assumption between treated and control groups. Each dot represents the expected mean outcome for the treated (red) and control (blue) groups. 
    Since these groups are a mix of respondents and nonrespondents, each quantity in post-treatment period is \textit{not} directly observed. (Middle): Parallel trends assumption between respondents and nonrespondents within the treated group. Each dot represents the expected mean outcome for the treated-respondent (red square) and treated-nonrespondent (red triangle) groups. The difference in the outcome of treated-respondents at the bottom is observed. (Right): Parallel trends assumption between respondents and nonrespondents within the control group. Each dot represents the expected mean outcome for the control-respondent (blue square) and control-nonrespondent (blue triangle) groups. The difference in the outcome of control-respondents at the bottom is observed.}
    \label{fig:illustration_cc}
\end{figure}

Under \Cref{a:pt-outcome} and \Cref{a:pt-observed}, the ATT can be identified by the DID estimand using complete case analysis.
The intuition is as follows:
As illustrated in the left panel of Figure \ref{fig:illustration_cc}, the canonical parallel trends assumption (\Cref{a:pt-outcome}) allows us to identify the ATT with the DID estimand.
Here, the red vertical line represents the difference in outcome for the treated group ($\E[Y_{i2} - Y_{i1} \mid D_{i} = 1]$), and the blue vertical line represents the difference in outcome for the control group ($\E[Y_{i2} - Y_{i1} \mid D_{i} = 0]$). 
By taking the difference between these two lines, we can identify the ATT as shown in Equation~\ref{eq:ident-att}.
Yet, due to the missingness in the outcome variable, these quantities are not directly observed.
Under \Cref{a:pt-observed}, however, we can identify each of these quantities by the difference in outcome of the treated-respondents and control-respondents, respectively, as shown in the middle and right panels of Figure \ref{fig:illustration_cc}.
For example, two red vertical lines in the middle panel represent the difference in outcome for the treated-respondents ($\E[Y_{i2} - Y_{i1} \mid D_{i} = 1, R_{i} = 1]$) and treated-nonrespondents ($\E[Y_{i2} - Y_{i1} \mid D_{i} = 1, R_{i} = 0]$), respectively, and are identical to the red vertical line in the left panel by \Cref{a:pt-observed}. 
Since the red vertical line at the bottom is observed ($\E[Y_{i2} - Y_{i1} \mid D_{i} = 1, R_{i} = 1]$), we can identify the difference in outcome for the treated group ($\E[Y_{i2} - Y_{i1} \mid D_{i} = 1]$). 
Similarly, we can identify that of the control group as well. 
I formally state this identification result in the following proposition.

\begin{namedproposition}{prop:cc-att}{Identification of ATT with complete case analysis}{
Under \namedCref{a:pt-outcome} and \namedCref{a:pt-observed}, we have
    \begin{equation*}
        \text{ATT} = \E[Y_{i2} - Y_{i1} \mid D_{i} = 1, R_{i} = 1] - \E[Y_{i2} - Y_{i1} \mid D_{i} = 0, R_{i} = 1]
    \end{equation*}}
\end{namedproposition}
\begin{proof}
    See Appendix~\ref{app:cc-att} for the proof.
\end{proof}
In simpler terms, if we adopt the canonical parallel trends assumption (\Cref{a:pt-outcome}), and additionally assume that there is no selection bias in the missingness of the before-after difference in the outcome variable within each treatment group (\Cref{a:pt-observed}), then the ATT can be identified using DID estimator with complete case analysis. As previously discussed, the parallel trends assumption between treated-respondents and treated-nonrespondents is more restrictive than the canonical parallel trends assumption. This is particularly the case when heterogeneous treatment effects related to missingness are present. In the subsequent section, we will delve further into the main identification challenges associated with missing data in the DID design.

\subsection{Identification Challenges with Missing Data}

In a general setup, when the parallel trends assumption is less plausible, the researcher may alternatively consider conditional parallel trends assumption where we assume parallel trends for each subgroup defined by pretreatment covariates. The rationale behind this is that the parallel trends assumption will become more plausible when the baseline covariates are balanced between the treated and control groups.
Given that the missingness is a post-treatment binary variable, we can consider the parallel trends assumptions conditioning on the principal strata, which can be viewed as a conditional parallel trends assumption, conditioning on a latent subgroup: never-respondents, if-treated-respondents, if-control-respondents, and always-respondents.

\begin{namedassumption}{a:pt-principal}{Principal strata parallel trends}{
    \begin{equation*}
        \E[Y_{i2}(0) - Y_{i1} \mid D_{i} = 0, S_{i} = s] = \E[Y_{i2}(0) - Y_{i1} \mid D_{i} = 1, S_{i} = s]
    \end{equation*}
    for $s \in \{(0,0), (0,1), (1,0), (1,1)\}$.}
\end{namedassumption}

This can be understood as a weaker assumption than the canonical parallel trends assumption (\Cref{a:pt-outcome}). It is worth pointing out that \Cref{a:pt-principal} does not imply \Cref{a:pt-outcome} in general, and vice versa.
\begin{namedremark}{r:pt}{Principal strata parallel trends does not imply canonical parallel trends}{
    \Cref{a:pt-principal} does not imply \Cref{a:pt-outcome}, and vice versa. One special case where \Cref{a:pt-principal} does imply \Cref{a:pt-outcome} is when the principal strata is independent of treatment selection, i.e. $S_{i} \indep D_{i}$.
    Under \Cref{a:pt-principal}, by the law of iterated expectation
    \begin{align*}
        \E[Y_{i2}(0) - Y_{i1} \mid D_{i} = 1] &=\sum_{s} \E[Y_{i2}(0) - Y_{i1} \mid D_{i} = 1, S_{i} = s] \Pr(S_{i} = s \mid D_{i} = 1) \\
        &=\sum_{s} \E[Y_{i2}(0) - Y_{i1} \mid D_{i} = 0, S_{i} = s] \Pr(S_{i} = s \mid D_{i} = 1)
    \end{align*}
    If $\Pr(S_{i} = s \mid D_{i} = 1) = \Pr(S_{i} = s \mid D_{i} = 0)$ then the last line is equal to $\E[Y_{i2}(0) - Y_{i1} \mid D_{i} = 0]$, which implies \Cref{a:pt-outcome}.}
\end{namedremark}

Intuitively speaking, even if the time trend of outcome is parallel between treated and control groups within a specific principal strata (e.g. always-respondents), unless the distribution of principal strata is equivalent between treated and control groups, the canonical parallel trends assumption may not hold. For example, in the second motivating application, if the incumbent supporters are more likely to respond to the questions on economic issues regardless of the partisan cues, the proportion of always-respondents may be higher in the treated group than the control group, which may potentially violate the canonical parallel trends assumption.

\begin{figure}[ht!]
    \centering
    \includegraphics[width = \linewidth]{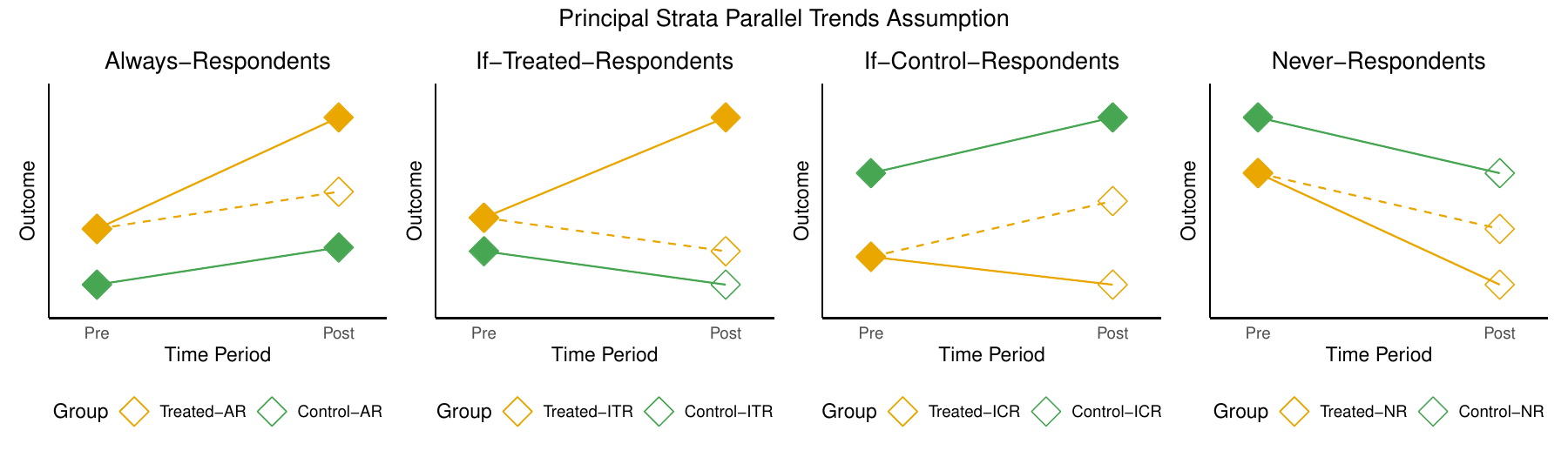}
    \caption{Toy Example of the Principal Strata Parallel Trends \Cref{a:pt-principal}. 
    Each dot represents the expected mean outcome for the treated (yellow) and control (green) groups within each principal strata.
    Filled dots represent the observed quantity and empty dots represent the unobserved quantity.
    Dashed lines represent the outcome trends under no administration of treatment for the treated group, i.e. line connecting $\E[Y_{i1} \mid D_{i} = 1, S_{i} = s]$ and $\E[Y_{i2}(0) \mid D_{i} = 1, S_{i} = s]$.
    Note that the time trend and causal effect are allowed to be heterogenous across principal strata.}
\end{figure}

Now, we discuss and clarify the main identification challenges with outcome MNAR using this principal strata parallel trends assumption.
\begin{namedproposition}{prop:decomp-att}{Decomposition of ATT with outcome MNAR}{
    Under \Cref{a:pt-principal}, we have
    \begin{align*}
        &\E[Y_{i2}(1) - Y_{i2}(0) \mid D_{i} = 1] \\
        &= \E[Y_{i2} - Y_{i1} \mid D_{i} = 1, R_{i} = 1] \Pr(R_{i} = 1 \mid D_{i} = 1) \\
        &\quad- \E[Y_{i2} - Y_{i1} \mid D_{i} = 0, \textcolor{dodgerblue}{R_{i}(1) = 1}, R_{i} = 1] \Pr(R_{i} = 1, \textcolor{dodgerblue}{R_{i}(0) = 1} \mid D_{i} = 1)\\
        &\quad- \E[Y_{i2} - Y_{i1} \mid D_{i} = 0, \textcolor{dodgerblue}{R_{i}(1) = 1}, R_{i} = 0] \Pr(R_{i} = 1, \textcolor{dodgerblue}{R_{i}(0) = 0} \mid D_{i} = 1) \\
        &\quad+ \E[\textcolor{indianred}{Y_{i2}} - Y_{i2}(0) \mid D_{i} = 1, R_{i} = 0, R_{i}(0) = 0] \Pr(R_{i} = 0, R_{i}(0) = 0 \mid D_{i} = 1) \\
        &\quad+ \E[\textcolor{indianred}{Y_{i2}} - Y_{i2}(0) \mid D_{i} = 1, R_{i} = 0, R_{i}(0) = 1] \Pr(R_{i} = 0, R_{i}(0) = 1 \mid D_{i} = 1)
    \end{align*}
    where the term in blue is not identified since principal strata is latent, and the term in red is not identified due to the missingness.
    This implies that the identification of ATT is not possible without further assumptions on:
    \begin{enumerate}
        \item Dependence of selection into treatment with principal strata
        \item Heterogeneous effect across principal strata
    \end{enumerate}
}
\end{namedproposition}
\begin{proof}
    See Appendix~\ref{app:decomp-att} for the proof.
\end{proof}

This decomposition is intuitive in the sense that it follows a natural logic starting from the respondents data and then correcting the potential bias due to the missingness. First, suppose that we naively compute $\E[Y_{i2} - Y_{i1} \mid D_{i} = 1, R_{i} = 1] \Pr(R_{i} = 1 \mid D_{i} = 1)$, which is before-after difference of treated-respondents weighted by the proportion of respondents. Then, we need to adjust for the time trend for the respondents: always-respondents and if-treated-respondents. This corresponds to the second and third lines, using the principal strata parallel trends assumption. Also, we need to incorporate the effect for the nonrespondents: never-respondents and if-control-respondents. This corresponds to the last two lines, which cannot be identified with the principal strata parallel trends assumption since the outcome is missing.

In other words, if we assume that the missingness pattern is independent of selection into treatment (i.e., $S_{i} \indep D_{i}$) \textit{and} that the treatment effect is homogeneous across principal strata within the treated group (i.e., $\E[Y_{i2}(1) - Y_{i2}(0) \mid D_{i} = 1, S_{i} = s]$ is constant for all $s$), then the ATT can be identified using the principal strata parallel trends assumption and complete case analysis as before. While these assumptions may be plausible in some applications, they are generally strong and may not hold. In particular, the assumption of a homogeneous treatment effect across principal strata within the treated group is restrictive and may be violated under missing not at random (MNAR).

In the subsequent section, we will propose an alternative approach to identify the ATT using auxiliary variables from the panel data without imposing the homogeneous treatment effect assumption across principal strata. Note that this can be viewed as MNAR in the missing data literature, where the missingness is dependent on the unobserved outcome variable.
Specifically, I first adapt an instrumental variable approach from \cite{dukes2022}, where the response indicator of the auxiliary variables is used as an instrumental variable for the missingness of the outcome variable and thus point identification of ATT is possible.
I also propose a partial identification approach based on \cite{lee2009} using pre-treatment missingness trend. This allows us to partially identify the ATT for always-respondents, in a more general setup of the dependence between the missingness and the treatment selection.

\section{The Proposed Methodology} \label{sec:proposed}

In this section, I propose two alternative approaches that do not require the homogeneous treatment effect assumption across principal strata. The first approach is based on the instrumental variable (IV) method, where the key idea is to utilize an IV that is associated with baseline missingness probability but not with the magnitude of the bias. I motivate this IV approach using randomized incentives for participation in the survey and also consider the response indicator of the auxiliary variables. The second approach is based on the partial identification method, where I tailor Lee bounds \citep{lee2009} to address the missing data problem in the DID design. Based on the parallel trends assumptions of response rate over time, I show that the ATT for always-respondents can be partially identified without requiring the homogeneous treatment effect assumption across principal strata or the independence assumption between missingness and treatment selection.

\subsection{Identification of ATT with Baseline Response Indicator} \label{sec:proposed-att}

We first consider the point identification of the ATT using an IV that captures the baseline probability of response under a canonical parallel trends assumption.
This approach is motivated by `bespoke instrument variable (IV)' from \cite{dukes2022}, in which they introduce a special type of IV that can be leveraged to identify the selection bias due to missingness \citep[also see][]{tchetgen2017, richardson2022}. 
Here, I adapt the basic idea of (1) introducing an IV that satisfies a certain exclusion restriction and (2) assuming that the bias is homogeneous across the subgroups defined by the IV. 
In \cite{dukes2022}, the focus was on a general framework for identifying selection bias due to missingness, specifically the difference $\E[Y_{i2} - Y_{i1} \mid R_{i} = 1] - \E[Y_{i2} - Y_{i1} \mid R_{i} = 0]$, under the assumption of a randomized experiment. However, this paper considers observational studies, where we employ parallel trends assumptions to identify the ATT. Our primary focus is on identifying selection bias due to missingness within each treatment group, i.e. $\delta_{d} \equiv \E[Y_{i2} - Y_{i1} \mid D_{i} = d, R_{i} = 1] - \E[Y_{i2} - Y_{i1} \mid D_{i} = d, R_{i} = 0]$ for $d = 0,1$.

To illustrate this IV approach, hypothetically assume that we offer a randomized incentive for survey participation. Let $\widetilde{R}_{i}$ denote the binary indicator of whether unit $i$ received the incentive. We consider the following set of assumptions on this indicator variable.

\begin{figure}[t!]
    \centering
    \begin{tikzpicture}
        \node (Rtilde) at (-2,-1) {$\widetilde{R}$};
        \node (R) at (0,-1) {$R$};
        \node (U) at (1,0) {$U$};
        \node (Y) at (2,-1) {$(Y_2 - Y_1)$};
    
        \draw[->, thick] (Rtilde) -- (R);
        \draw[->, thick, dotted] (U) -- (R);
        \draw[->, thick, dotted] (U) -- (Y);
    \end{tikzpicture}
    \caption{A DAG example of Randomized Incentive ($\widetilde{R}$) and Post-treatment Response Indicator ($R$). Omitted $D$ for simplicity.}
\end{figure}
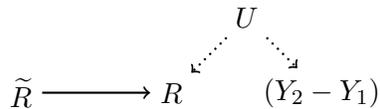

\begin{namedassumption}{a:missing_iv}{IV approach}{
    An indicator variable $\widetilde{R}_{i} \in \{0,1\}$ satisfies the following set of assumptions:
    \begin{enumerate}
        \item \textbf{Relevance to missingness}: $\widetilde{R}_{i}$ is relevant to the missingness of the post-treatment outcome.
        \begin{equation*}
            \Pr(R_{i} = 0 \mid D_{i} = d, \widetilde{R}_{i} = 0) \neq \Pr(R_{i} = 0 \mid D_{i} = d, \widetilde{R}_{i} = 1) 
        \end{equation*}
        for $d = 0,1$.
        \item \textbf{Parallel trends of observed outcome}: The observed trend of the outcome is parallel between two groups defined by $\widetilde{R}_{i}$.
        \begin{equation*}
            \E[Y_{i2} - Y_{i1} \mid D_i = d, \widetilde{R}_{i} = 0] = \E[Y_{i2} - Y_{i1} \mid D_i = d, \widetilde{R}_{i} = 1]
        \end{equation*}
        for $d = 0,1$.
        \item \textbf{Bias homogeneity}: The bias due to post-treatment missingness in the time trend of the outcome is homogeneous between two groups defined by $\widetilde{R}_{i}$.
        \begin{align*}
            \delta_{d}
            &=\E[Y_{i2} - Y_{i1} \mid D_i = d, \widetilde{R}_{i} = 0, R_{i} = 1] - \E[Y_{i2} - Y_{i1} \mid D_i = d, \widetilde{R}_{i} = 0, R_{i} = 0] \\
            &=\E[Y_{i2} - Y_{i1} \mid D_i = d, \widetilde{R}_{i} = 1, R_{i} = 1] - \E[Y_{i2} - Y_{i1} \mid D_i = d, \widetilde{R}_{i} = 1, R_{i} = 0]
        \end{align*}
        where 
        \begin{equation*}
            \delta_{d} \equiv \E[Y_{i2} - Y_{i1} \mid D_i = d, R_{i} = 1] - \E[Y_{i2} - Y_{i1} \mid D_i = d, R_{i} = 0]
        \end{equation*}
        for $d = 0,1$.
    \end{enumerate}
}
\end{namedassumption}

The first assumption, \textbf{relevance to missingness} applies if $\widetilde{R}_{i}$ is correlated with the missingness of the post-treatment outcome. 
For example, if the incentive encourages the respondents to participate in the survey, then the missingness of the outcome may be negatively correlated with the receipt of the incentive.
Intuitively speaking, if we consider the post-treatment missingness to be a combination of MAR (in terms of baseline response probability) and MNAR mechanisms, $\widetilde{R}_{i}$ helps us control for the MAR component of the missingness.
Unlike others, this assumption can be empirically tested. 

The second assumption, \textbf{parallel trends of observed outcome}, corresponds to \namedCref{a:pt-observed} with regards to $\widetilde{R}_{i}$ instead of $R_{i}$. As discussed in a previous section, this assumption can be violated if $\widetilde{R}_{i}$ is correlated with the treatment effect size while the time trend remains constant. This suggests that a crucial criterion for $\widetilde{R}_{i}$ as an IV is the expectation of homogeneous effects between groups defined by this indicator. 
Under the incentive scenario, this assumption is plausible since the incentive is randomized.

The last assumption, \textbf{bias homogeneity}, is the core assumption for using $\widetilde{R}_{i}$ to correct the bias. This assumption holds if the bias resulting from post-treatment missingness is consistent across individuals who received the incentive and those who did not, within each treatment group. This assumption might be violated if there is an unmeasured confounder, denoted as $U_{i}$, that interacts with the outcome and $\widetilde{R}_{i}$. For instance, in an additive outcome model, the presence of an interaction term $U_{i} \times \widetilde{R}_{i}$ suggests that the bias in post-treatment outcomes between two groups defined by $\widetilde{R}_{i}$ could be heterogeneous. 

Intuitively speaking, when we consider all these assumptions together, we can deduce that any observed differences between the $\{\widetilde{R}_{i} = 0\}$ and $\{\widetilde{R}_{i} = 1\}$ groups are due to their association with post-treatment missingness, which induces bias. 
The identification result is formally stated in the following theorem.

\begin{namedtheorem}{thm:att-iv}{Identification of ATT using IV}{
    Under \Cref{a:pt-outcome} and \ref{a:missing_iv}, the ATT can be identified using $\widetilde{R}_{i}$ as an IV.
        \begin{align*}
            &\text{ATT} \\
            &= \E[Y_{i2} - Y_{i1} \mid D_{i} = 1, R_{i} = 1] - \E[Y_{i2} - Y_{i1} \mid D_{i} = 0, R_{i} = 1]  \\
            &\quad + \frac{\E[Y_{i2} - Y_{i1} \mid D_{i} = 1, \widetilde{R}_{i} = 1, R_{i} = 1] - \E[Y_{i2} - Y_{i1} \mid D_{i} = 1, \widetilde{R}_{i} = 0, R_{i} = 1]}{\Pr(R_{i}  = 0 \mid D_{i} = 1, \widetilde{R}_{i} = 0) - \Pr(R_{i}  = 0 \mid D_{i} = 1, \widetilde{R}_{i} = 1)} \\
            &\qquad\qquad \times \Pr(R_{i} =0 \mid D_{i}=1)  \\
            &\quad - \frac{\E[Y_{i2} - Y_{i1} \mid D_{i} = 0, \widetilde{R}_{i} = 1, R_{i} = 1] - \E[Y_{i2} - Y_{i1} \mid D_{i} = 0, \widetilde{R}_{i} = 0, R_{i} = 1]}{\Pr(R_{i}  = 0 \mid D_{i} = 0, \widetilde{R}_{i} = 0) - \Pr(R_{i}  = 0 \mid D_{i} = 0, \widetilde{R}_{i} = 1)} \\
            &\qquad\qquad \times \Pr(R_{i} =0 \mid D_{i}=0)
        \end{align*}
    }
    \end{namedtheorem}
        
    \begin{proof}
        Proof is straightforward by Lemma~\ref{lem:inst} in Appendix~\ref{app:proposed-att}.
    \end{proof}
    
So far, I have used a randomized incentive for survey participation as a potential choice of an IV. 
Alternatively, one might consider using the response indicator of auxiliary variables as an IV. 
These auxiliary variables can be constructed from other pre-treatment variables, such as those survey questions unrelated to the main outcome of interest.
For instance, in the first motivating example, we could use other outcome variables (e.g., ``Afghanistan right direction?'') from the 2016 wave as auxiliary variables. However, it is crucial to ensure that the auxiliary variable satisfies the assumptions in \Cref{a:missing_iv}.
For example, in the initial motivating application, if ``local government confidence'' is the primary outcome of interest and there is concern that the missingness of the auxiliary variable correlates with the treatment effect size, it would be advisable to choose another outcome variable that is not directly associated with local governance.

In Appendix~\ref{app:proposed-att}, I generalize the setup to allow for missingness in pre-treatment periods and discuss the identification of the ATT using the IV approach. 
In Appendix~\ref{app:multi}, I also present a variant of the IV method discussed above, where the assumption of parallel trends in the observed outcome for the auxiliary variable is relaxed. Instead, this variant employs multiple auxiliary variables that exhibit consistent bias between respondents and nonrespondents. This approach offers a more flexible framework, accommodating scenarios where the missingness of the auxiliary variables may correlate with the treatment effect size.

Despite its potential, the proposed IV approach has several limitations. First and foremost, the bias homogeneity assumption is crucial for the identification of the ATT, yet its validity is difficult to verify.
\cite{tchetgen2017} provides an example of semiparametric shared parameter model that satisfies the bias homogeneity assumption. 
A future research direction could be to show which types of models in our setup satisfy this assumption. 
For example, a separable model of the post-treatment missingness in terms of $\widetilde{R}_{i}$ and an unmeasured confounder $U_{i}$ could be a potential candidate
\footnote{For instance, consider the following semiparametric model:
\begin{align*}
    \E[Y_{2} - Y_{1} \mid  U, D] &= f(D) + U\\
    \text{logit} \Pr (R = 1 \mid \widetilde{R}, U, D) &= \underbrace{g_1(D,\widetilde{R}) + g_2(D) U}_{\text{A separable model}}\\
    \text{logit} \Pr (\widetilde{R} = 1 \mid \widetilde{U}, D) &= h(D)\widetilde{U}\\
    U \mid R = 0, \widetilde{R}, D &\sim m(D,\widetilde{R}) + \zeta
\end{align*}}.
Alternatively, as mentioned in \cite{tchetgen2017}, one may consider a hypothesis test of $\E[Y_{i2} - Y_{i1} \mid D_i = d, \widetilde{R}_{i} = 1, R_{i} = 1] - \E[Y_{i2} - Y_{i1} \mid D_i = d, \widetilde{R}_{i} = 1, R_{i} = 0] = 0$ without making such bias homogeneity assumption.

More importantly, we apply the canonical parallel trends assumption (\Cref{a:pt-outcome}) in this approach instead of the principal strata parallel trends assumption (\Cref{a:pt-principal}).
It's crucial to highlight that, should we choose to adopt assumption \Cref{a:pt-principal} in lieu of \Cref{a:pt-outcome}, it may be required to introduce an additional assumption such as the selection into treatment is independent of the principal strata (see Remark~\ref{r:mist-ind} for more details). In the subsequent section, I modify this assumption by allowing the selection into treatment to depend on the principal strata, and propose a partial identification approach to estimate the ATT for the subgroup of always-respondents. 

\subsection{Partial Identification of ATT for Always-Respondents with Response Parallel Trends} \label{sec:proposed-partial}

In this section, I introduce a partial identification of the ATT specifically for the subgroup of always-respondents, leveraging the trend observed in pre-treatment missingness. The partial identification follows the `trimming bounds' approach of \cite{zhang2003} and \cite{lee2009}, where the bounds are derived by considering the extreme cases based on the observed distribution.
Here, we are interested in the following quantity:
\begin{equation*}
    \text{ATT-AR} = \E[Y_{i2}(1) - Y_{i2}(0) \mid D_{i} = 1, R_{i2} (1) = 1, R_{i2} (0) = 1]
\end{equation*}
ATT-AR is an average treatment effect among always-respondents who are selected into the treatment. 
As discussed in the previous section, if researchers believe that the missingness is MNAR, but are reluctant to make additional assumptions about effect heterogeneity across principal strata, then the ATT-AR is the only quantity that can be partially identified with minimal assumptions including \namedCref{a:pt-principal}. When the primary interest is in the ATT, this approach serves as a valuable starting point to understand the potential biases introduced by missingness and to develop a corresponding sensitivity analysis.

Since we utilize pre-treatment missingness in this approach, we slightly modify the notation to distinguish the missingness indicator at different time points. Specifically, we use $R_{it}$ instead of $R_{i}$ to denote the response indicator of the outcome variable at time $t$. That is, $R_{it} = 1$ if $Y_{it}$ is observed, and $0$ if it is missing, for $t = 1,2$. We also assume that we are interested in the ATT for the population who responded to the survey at time 1, while omitting $R_{i1} = 1$ in the conditioning set for simplicity of notation. Alternatively, we could consider an exclusion restriction type of assumption for pre-treatment missingness and time trends of the outcome, as described in \Cref{a:pt_mar} in \Cref{app:proposed-att}.

\paragraph{Trimming bounds.}
We first review the main idea of the trimming bounds approach.
Let $\pi_{r_{1}, r_{0}}(d) \equiv \Pr(R_{i2}(1) = r_{1}, R_{i2}(0) = r_{0} \mid D_{i} = d)$.
For now, suppose that we identified the proportion of principal strata in the treated group. We will discuss more about this in the later part.

Under \namedCref{a:pt-principal}, we have
\begin{equation*}
    \text{ATT-AR} = \E[Y_{i2} - Y_{i1} \mid D_{i} = 1, R_{i2} (1) = 1, R_{i2} (0) = 1] + \E[Y_{i2} -  Y_{i1} \mid D_{i} = 0, R_{i2} (1) = 1, R_{i2} (0) = 1]
\end{equation*}
where the right hand side is the DID among always respondents. Our goal here is to bound this quantity using the observed data:
\begin{align*}
    f_{1}(y) &\equiv \Pr(Y_{i2} - Y_{i1} \le y \mid D_{i} = 1, R_{i2}  = 1) \\
    &= \frac{\pi_{11}(1)}{\pi_{11}(1) + \pi_{10}(1)} \Pr(Y_{i2} - Y_{i1} \le y \mid D_{i} = 1, R_{i2} (1) = 1, R_{i2} (0) = 1) \\
    &\qquad+ \frac{\pi_{10}(1)}{\pi_{11}(1) + \pi_{10}(1)} \Pr(Y_{i2} - Y_{i1} \le y \mid D_{i} = 1, R_{i2} (1) = 1, R_{i2} (0) = 0) 
\end{align*}
Suppose $\frac{\pi_{11}(1)}{\pi_{11}(1) + \pi_{10}(1)} = p/100$.
Following `trimming bounds' approach of \cite{zhang2003} and \cite{lee2009}, we can get the lower bound of $\E[Y_{i2} - Y_{i1} \mid D_{i} = 1, R_{i2} (1) = 1, R_{i2} (0) = 1]$ by `trimming' lower $p\%$ of treated-respondent group and taking average. That is,
\begin{equation*}
    \E[Y_{i2} - Y_{i1} \mid D_{i} = 1, R_{i2} (1) = 1, R_{i2} (0) = 1] \ge \int_{-\infty}^{q^{\text{low}}_{1}(f_{1})} y \frac{f_{1}(y)}{\int_{-\infty}^{q^{\text{low}}_{1}(f_{1})} f_{1}(y) dy} dy
\end{equation*}
where $q^{\text{low}}_{1}(f_{1})$ is the $\frac{\pi_{10}(1)}{\pi_{11}(1) + \pi_{10}(1)}$ quantile of $f_{1}(\cdot)$. Similarly, we can get the upper bound as follows:
\begin{equation*}
    \E[Y_{i2} - Y_{i1} \mid D_{i} = 1, R_{i2} (1) = 1, R_{i2} (0) = 1] \le \int_{q^{\text{high}}_{1}(f_{1})}^{\infty} y \frac{f_{1}(y)}{\int_{q^{\text{high}}_{1}(f_{1})}^{\infty} f_{1}(y) dy} dy
\end{equation*}
where $q^{\text{high}}_{1}(f_{1})$ is the $\frac{\pi_{11}(1)}{\pi_{11}(1) + \pi_{10}(1)}$ quantile of $f_{1}(\cdot)$.

\begin{figure}[ht!]
    \centering
    \includegraphics[width = \linewidth]{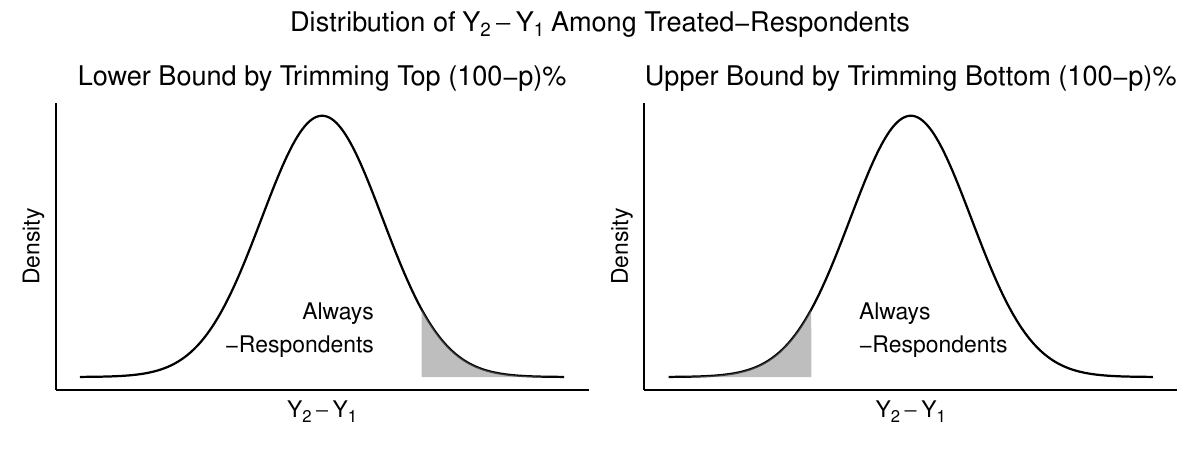}
    \caption{Illustration of the trimming bounds approach for the ATT-AR. The un-shaded area represents the lower and upper bounds of $\E[Y_{i2} - Y_{i1} \mid D_{i} = 1, R_{i2} (1) = 1, R_{i2} (0) = 1]$.}
\end{figure}

Using the same strategy, we can bound $\E[Y_{i2} - Y_{i1} \mid D_{i} = 0, R_{i2} (1) = 1, R_{i2} (0) = 1]$:
\begin{align*}
    \E[Y_{i2} - Y_{i1} \mid D_{i} = 0, R_{i2} (1) = 1, R_{i2} (0) = 1] &\ge \int_{-\infty}^{q^{\text{low}}_{0}(f_{0})} y \frac{f_{0}(y)}{\int_{-\infty}^{q^{\text{low}}_{0}(f_{0})} f_{0}(y) dy} dy \\
    \E[Y_{i2} - Y_{i1} \mid D_{i} = 0, R_{i2} (1) = 1, R_{i2} (0) = 1] &\le \int_{q^{\text{high}}_{0}(f_{0})}^{\infty} y \frac{f_{0}(y)}{\int_{q^{\text{high}}_{0}(f_{0})}^{\infty} f_{0}(y) dy} dy
\end{align*}
where $f_{0} \equiv \Pr(Y_{i2} - Y_{i1} \le y \mid D_{i} = 0, R_{i2}  = 1)$, $q^{\text{low}}_{0}(f_{0})$ is the $\frac{\pi_{01}(0)}{\pi_{11}(0) + \pi_{01}(0)}$ quantile of $f_{0}(\cdot)$, $q^{\text{high}}_{0}(f_{0})$ is the $\frac{\pi_{11}(0)}{\pi_{11}(0) + \pi_{01}(0)}$ quantile of $f_{0}(\cdot)$.
Consequently, our remaining task involves identifying the proportion of principal strata within the treated group. This analysis diverges from the standard principal strata analysis, which typically focuses on randomized experiments. Here, we are delving into observational studies and employing a DID design instead of relying on the assumption of ignorability (refer to \citealp{wang2017}). In this paper, we introduce the parallel trends assumption in the missingness of outcomes as a strategy to address the challenges associated with identifying the proportion of principal strata in the treated group.

\paragraph{Principal strata proportion.}
In this paper, we consider two approaches for identifying the principal strata proportion within the treated group: one that incorporates a monotonicity assumption and another that does not. We first start with the method that assumes monotonicity alongside parallel trends in the response indicator.

\begin{namedassumption}{a:monotonicity}{Monotonicity}{    
    \begin{equation*}
        R_{i2}(1) \geq R_{i2}(0), \quad \forall i
    \end{equation*}}
\end{namedassumption}
The monotonicity assumption suggests that the treatment effect on the response indicator for any individual is always non-negative. This assumption is reasonable in certain contexts. For example, consider our first motivating study where we are interested in the effect of aid on the public perception toward the government. If this aid positively affects individuals' perceptions, it is plausible to expect a higher response rate in the post-treatment survey under treatment compared to the control. This is based on the premise that individuals with negative perceptions are more hesitant to participate in the survey. Thus, a positive treatment effect on perception would likely increase the likelihood of response, reflecting the monotonicity stated above.
With this assumption, we have 
\begin{align*}
    \Pr(R_{i2} (1) = 1, R_{i2} (0) = 1 \mid D_{i} = 1) &= \Pr(R_{i2} (0) = 1 \mid D_{i} = 1) \\
    \Pr(R_{i2} (1) = 1, R_{i2} (0) = 1 \mid D_{i} = 0) &= \Pr(R_{i2} (0) = 1 \mid D_{i} = 0) = \Pr(R_{i2}  = 1 \mid D_{i} = 0)
\end{align*}
Thus, we only need to identify $\Pr(R_{i2} (0) = 1 \mid D_{i} = 1)$ for always respondents in the treated group. To this end, we introduce the following assumption.

\begin{namedassumption}{a:pt-missing}{Parallel trends of missingness}{
    \begin{equation*}
        \E[R_{i2}(0) - R_{i1}(0) \mid D_{i} = 1] = \E[R_{i2}(0) - R_{i1}(0) \mid D_{i} = 0]
    \end{equation*}}
\end{namedassumption}

With this assumption, we can identify $\Pr(R_{i2} (0) = 1 \mid D_i = 1) = \Pr(R_{i2} = 1 \mid D_{i} = 0) - \Pr(R_{i1} = 1 \mid D_{i} = 0) + \Pr(R_{i1}  = 1 \mid D_{i} = 1)$.

\begin{namedproposition}{prop:monotonicity}{Identification of principal strata proportion with monotonicity}{
    Under \Cref{a:monotonicity} and \ref{a:pt-missing}, we can identify the following:
    \begin{align*}
        \pi_{11}(1) &= \Pr(R_{i2}  = 1 \mid D_{i} = 0) - \Pr(R_{i1}  = 1 \mid D_{i} = 0) \\
        &\qquad + \Pr(R_{i1}  = 1 \mid D_{i} = 1)\\
        \pi_{10}(1) &= \Pr(R_{i2}  = 1 \mid D_i = 1) - \Pr(R_{i2}  = 1 \mid D_{i} = 0) \\
        &\qquad - \{\Pr(R_{i1}  = 1 \mid D_{i} = 1) - \Pr(R_{i1}  = 1 \mid D_{i} = 0)\} \\
        \pi_{11}(0) &= \Pr(R_{i2}  = 1 \mid D_{i} = 0) \\
        \pi_{01}(0) &= 0
    \end{align*}}
\end{namedproposition}

In certain scenarios, the monotonicity assumption might not be applicable. Take, for instance, our second motivating example where the treatment effect—altering cues may not uniformly result in an increased response rate in the post-treatment survey. In such instances, one may consider the following assumption.

\begin{namedassumption}{a:homogeneous}{Equivalence of ATT and ATC on missingness}{
    \begin{equation*}
        \E[R_{i2} (1) - R_{i2} (0) \mid D_{i} = 1] = \E[R_{i2} (1) - R_{i2} (0) \mid D_{i} = 0]
    \end{equation*}}
\end{namedassumption}

\begin{namedproposition}{prop:bound}{Identification of counterfactual response without monotonicity}{
    Under \Cref{a:pt-missing} and \ref{a:homogeneous} we have:
    \begin{align*}
        \Pr(R_{i2}(0) = 1 \mid D_i = 1) &= \Pr(R_{i2} = 1 \mid D_{i} = 0) - \Pr(R_{i1} = 1 \mid D_{i} = 0) + \Pr(R_{i1} = 1 \mid D_{i} = 1) \\
        \Pr(R_{i2}(1) = 1 \mid D_i = 0) &= \Pr(R_{i2} = 1 \mid D_{i} = 1) - \Pr(R_{i1} = 1 \mid D_{i} = 1) + \Pr(R_{i1} = 1 \mid D_{i} = 0)
    \end{align*}}
\end{namedproposition}

\paragraph{Partial identification of ATT-AR.}
Combining these identification results, we can bound ATT-AR.
\begin{namedtheorem}{thm:bound-attar}{Partial identification of ATT-AR}{
Under \namedCref{a:pt-principal} and \namedCref{a:pt-missing}, ATT-AR is partially identified:
\begin{enumerate}
    \item With \namedCref{a:monotonicity},
    \begin{equation*}
        \lb_{y_{2}-y_{1},1} - \int_{-\infty}^{\infty}y f_{0}(y)dy \le \text{ATT-AR} \le \ub_{y_{2}-y_{1},1} - \int_{-\infty}^{\infty}y f_{0}(y)dy
    \end{equation*}
    \item Without \namedCref{a:monotonicity}, but instead with \namedCref{a:homogeneous}
    \begin{equation*}
        \lb_{y_{2}-y_{1},1} - \ub_{y_{2}-y_{1},0} \le \text{ATT-AR} \le \ub_{y_{2}-y_{1},1} - \lb_{y_{2}-y_{1},0}
    \end{equation*}
\end{enumerate}
    where $\lb_{y_{2}-y_{1},d}$ and $\ub_{y_{2}-y_{1},d}$ are the lower and upper bounds of $\E[Y_{i2} - Y_{i1} \mid D_{i} = d, R_{i2}(1) = 1, R_{i2}(0) = 1]$.
    \begin{align*}
        &\lb_{y_{2}-y_{1},d} = \int_{-\infty}^{q^{\text{low}}_{d}(f_{d})} y \frac{f_{d}(y)}{\int_{-\infty}^{q^{\text{low}}_{d}(f_{d})} f_{d}(y) dy} dy \\
        &\ub_{y_{2}-y_{1},d} = \int_{q^{\text{high}}_{d}(f_{d})}^{\infty} y \frac{f_{d}(y)}{\int_{q^{\text{high}}_{d}(f_{d})}^{\infty} f_{d}(y) dy} dy
    \end{align*}
    where $f_{d}(y) \equiv \Pr(Y_{i2}-Y_{i1} \le y \mid D_{i} = d, R_{i2} = 1)$, $q^{\text{low}}_{d}(f_{d})$ is bottom $\frac{\pi_{11}(d)}{\Pr(R_{i2} = 1 \mid D_i = d)}$ quantile of $f_{d}(\cdot)$ and $q^{\text{high}}_{d}(f_{d})$ is top $\frac{\pi_{11}(d)}{\Pr(R_{i2} = 1 \mid D_i = d)}$ quantile of $f_{d}(\cdot)$.
}
\end{namedtheorem}

As noted earlier, the results in \Cref{thm:bound-attar} provide bounds on the ATT-AR, accommodating both the dependence between treatment selection and principal strata, as well as heterogeneous treatment effects across principal strata. This approach is particularly useful when the missingness rate is moderate, raising concerns about potential bias in the ATT while the proportion of always-respondents remains substantial. 

\section{Empirical Application} \label{sec:empirical}

In this section, I revisit the study by \cite{sexton2023} to illustrate the proposed method. 
The paper examines the impact of small development aid projects on public perception and attitudes toward the government. The study employs a DID design, aggregating data at the village cluster level to address inferential challenges (``spatial spillovers and incorrect standard errors''). The aim here is to demonstrate the application of the proposed method rather than to replicate the original study in its entirety. Therefore, for the purpose of this illustration, I simplify the data structure and employ a canonical two-period DID estimator.\footnote{The results presented in this section are provided for illustrative purposes only and should not be interpreted as supporting substantive conclusions.}

\begin{table}[ht!]
    \centering
    \begin{tabular}{lrrrrrrll}
        \toprule
        \multicolumn{1}{c}{ } & \multicolumn{2}{c}{$\pi_{11}(1)$} & \multicolumn{2}{c}{$\pi_{11}(0)$} & \multicolumn{2}{c}{ATT-AR} & \multicolumn{1}{c}{ } \\
        \cmidrule(l{2pt}r{2pt}){2-3} \cmidrule(l{2pt}r{2pt}){4-5} \cmidrule(l{2pt}r{2pt}){6-7}
        Outcome & LB & UB & LB & UB & LB & UB & DID$^{1}$ & DID$^{2}$\\
        \midrule
        Afghanistan Right Direction? & 0.83 & 0.92 & 0.84 & 0.92 & -0.24 & 0.15 & 0.05 (0.04) & -0.04\\
        Confidence in President & 0.97 & 0.98 & 0.97 & 0.98 & 0.03 & 0.09 & -0.18 (0.13) & 0.05\\
        National Gov. Good Job & 0.99 & 0.99 & 0.98 & 0.99 & 0.03 & 0.05 & -0.2 (0.12) & 0.04\\
        Local Gov. Confidence & 0.00 & 0.34 & 0.72 & 0.85 & -1 & 1 & -0.29 (0.12) & -0.04\\
        Sympathy for Insurgents & 0.92 & 0.94 & 0.97 & 0.97 & -0.02 & 0.10 & 0.11 (0.04) & 0.04\\
        \bottomrule
        \end{tabular}
    \caption{Estimates of Bounds of ATT-AR following Theorem~\ref{thm:bound-attar}. LB = lower bound, UB = upper bound, $\pi_{11}(1)$ = proportion of always-respondents in the treated group, $\pi_{11}(0)$ = proportion of always-respondents in the control group, DID$^{1}$ = the DID estimates (clustered by village) from the original study with standard errors in parentheses, DID$^{2}$ = two-period DID estimates.}
    \label{tab:result2}
\end{table}

Table~\ref{tab:result2} shows the partial identification of ATT-AR for the five outcome variables considered in the original study. The table presents the lower and upper bounds of ATT-AR, along with the DID estimates. There are two main takeaways here: across the outcome variables, except for ``Local Government Confidence,'' the proportion of always-respondents is about $0.8$ to $0.9$ in both treatment and control groups. This result is valuable for researchers as ATT-AR provides a useful reference for ATT with the correction of bias, given a large proportion of this latent subgroup. Particularly for ``Confidence in President'' and ``National Government Good Job,'' where the missingness ratio was relatively small (see Table~\ref{tab:sexton_missing}), the proportion of always-respondents is close to $1$. In contrast, for ``Local Government Confidence,'' the proportion of always-respondents in the control group is about $0.8$, while it is close to $0$ in the treatment group. This is intuitive, given the fact that the missingness rate is about $70\%$ and $65\%$ for the treated group in pre- and post-treatment surveys, respectively. Secondly, the bounds of ATT-AR are comparatively tight, except for ``Local Government Confidence,'' where the bounds range from $-1$ to $1$, given the large gap between the treated and control groups.

\begin{table}[ht!]
    \centering
    \begin{tabular}{llll}
    \toprule
    Outcome & Bias Corrected DID & DID$^{1}$ & DID$^{2}$\\
    \midrule
    Confidence in President & 0.57 & -0.18 (0.13) & -0.04\\
    National Gov. Good Job & 0.24 & -0.2 (0.12) & 0.00\\
    Sympathy for Insurgents & 0.94 & 0.11 (0.04) & -0.05\\
    \bottomrule
    \end{tabular}
    \caption{Estimates of ATT following Theorem~\ref{thm:att-iv}. Bias Corrected DID = estimates using response indicator of ``Employment opportunity (self-reported)'' from 2016 (pre-treatment) wave as IV, DID$^{1}$ = the DID estimates (clustered by village) from the original study with standard errors in parentheses, DID$^{2}$ = two-period DID estimates.}
    \label{tab:result1}
\end{table}

Table~\ref{tab:result1} presents the ATT estimates obtained using the IV method outlined in Theorem~\ref{thm:att-iv}. These estimates were derived using the self-reported ``Employment opportunity'' response indicator from the 2016 (pre-treatment) wave as an instrumental variable. The findings indicate that the difference between the bias-corrected DID estimates and the standard DID estimates is proportional to the missingness ratio observed in the post-treatment survey.

\section{Concluding Remarks}

Missingness in panel data is a prevalent issue that has not received sufficient attention in the context of DID studies. In this paper, I address the challenges posed by missing data in DID analyses by exploring different variants of the parallel trends assumption and additional sources of information, such as the trend of missingness rates. Employing the principal strata framework, I identify two primary challenges when outcomes are Missing Not At Random (MNAR): (1) the potential dependence of treatment selection on principal strata and (2) heterogeneous effects across principal strata. To overcome these challenges, I propose two alternative strategies with weaker assumptions: (1) partial identification of the ATT for always-respondents and (2) an IV method that employs baseline missingness indicators as instruments.

The paper also suggests several avenues for extending these methods. For example, the nonparametric identification of ATT-AR could be achieved by adopting the assumptions from \cite{wang2017}. Additionally, integrating the principal strata framework into the IV approach could provide a way to account for the dependency between treatment selection and principal strata. Exploring partial identification through the bracketing relationship among multiple auxiliary variables, as inspired by \cite{ye2023}, offers another potential extension. Moreover, generalizing the approach to accommodate staggered adoption and multiple treatment groups represents a promising direction for future research. Lastly, addressing missingness in covariates within DID studies is also worth investigating (see Appendix~\ref{app:pi} for a relevant discussion).

\clearpage
\bibliographystyle{apsr}
\bibliography{did}

\clearpage
\appendix
\section{Proofs of Section~\ref{sec:standard}} \label{app:standard}
\subsection{Proof of Proposition~\ref{prop:cc-att}} \label{app:cc-att}
\begin{proof}
    Under consistency assumption and parallel trends assumption, we have
    \begin{align*}
        &\E[Y_{i2}(1) - Y_{i2}(0) \mid D_{i} = 1] \\
        &= \E[Y_{i2} - Y_{i1} \mid D_{i} = 1] - \E[Y_{i2} - Y_{i1} \mid D_{i} = 0] \\
        &= \sum_{r = 0}^{1} \E[Y_{i2} - Y_{i1} \mid D_{i} = 1, R_{i} = r] \Pr(R_{i} = r \mid D_{i} = 1) \\
        &\quad - \sum_{r = 0}^{1} \E[Y_{i2} - Y_{i1} \mid D_{i} = 0, R_{i} = r] \Pr(R_{i} = r \mid D_{i} = 0) \\
        &= \sum_{r = 0}^{1} \E[Y_{i2} - Y_{i1} \mid D_{i} = 1, R_{i} = 1] \Pr(R_{i} = r \mid D_{i} = 1) \\
        &\quad - \sum_{r = 0}^{1} \E[Y_{i2} - Y_{i1} \mid D_{i} = 0, R_{i} = 1] \Pr(R_{i} = r \mid D_{i} = 0) \\
        &= \E[Y_{i2} - Y_{i1} \mid D_{i} = 1, R_{i} = 1] - \E[Y_{i2} - Y_{i1} \mid D_{i} = 0, R_{i} = 1]
    \end{align*}
    where the second equality uses the law of total expectation and the third equality uses \Cref{a:pt-observed}.
\end{proof}
\subsection{Proof of Proposition~\ref{prop:decomp-att}} \label{app:decomp-att}
\begin{proof}
    By the law of total expectation,
\begin{align*}
    &\E[Y_{i2}(1) - Y_{i2}(0) \mid D_{i} = 1] \\
    &= \E[Y_{i2}(1) - Y_{i2}(0) \mid D_{i} = 1, R_{i} = 1] \Pr(R_{i} = 1 \mid D_{i} = 1) \\
    &+ \E[Y_{i2}(1) - Y_{i2}(0) \mid D_{i} = 1, R_{i} = 0] \Pr(R_{i} = 0 \mid D_{i} = 1) 
\end{align*}
Introducing the pre-treatment period outcome $Y_{i1}$, we can rewrite the above as
\begin{align*}
    &= \E[Y_{i2}(1) - Y_{i1} \mid D_{i} = 1, R_{i} = 1] \Pr(R_{i} = 1 \mid D_{i} = 1) \\
    &+ \E[Y_{i2}(1) - Y_{i1} \mid D_{i} = 1, R_{i} = 0] \Pr(R_{i} = 0 \mid D_{i} = 1) \\
    &- \E[Y_{i2}(0) - Y_{i1} \mid D_{i} = 1, R_{i} = 1] \Pr(R_{i} = 1 \mid D_{i} = 1) \\
    &- \E[Y_{i2}(0) - Y_{i1} \mid D_{i} = 1, R_{i} = 0] \Pr(R_{i} = 0 \mid D_{i} = 1)
\end{align*}
By the consistency assumption, we have
\begin{align}
    &= \E[Y_{i2} - Y_{i1} \mid D_{i} = 1, R_{i} = 1] \Pr(R_{i} = 1 \mid D_{i} = 1) \label{eq:ident_start}\\
    &+ \underbrace{\E[\textcolor{indianred}{Y_{i2}} - Y_{i1} \mid D_{i} = 1, R_{i} = 0] \Pr(R_{i} = 0 \mid D_{i} = 1)}_{\text{Part 2}} \\
    &- \underbrace{\E[Y_{i2}(0) - Y_{i1} \mid D_{i} = 1, R_{i} = 1] \Pr(R_{i} = 1 \mid D_{i} = 1)}_{\text{Part 1}} \\
    &- \underbrace{\E[Y_{i2}(0) - Y_{i1} \mid D_{i} = 1, R_{i} = 0] \Pr(R_{i} = 0 \mid D_{i} = 1)}_{\text{Part 1}} \label{eq:ident_end}
\end{align}
where the term in red is not identified due to the missingness.
\paragraph{Part 1} 
Again, by the law of total expectation and the principal strata parallel trends assumption,
\begin{align}
    &\E[Y_{i2}(0) - Y_{i1} \mid D_{i} = 1, R_{i} = 1] \Pr(R_{i} = 1 \mid D_{i} = 1) \nonumber\\
    &\quad+ \E[Y_{i2}(0) - Y_{i1} \mid D_{i} = 1, R_{i} = 0] \Pr(R_{i} = 0 \mid D_{i} = 1) \nonumber\\
    &= \E[Y_{i2}(0) - Y_{i1} \mid D_{i} = 0, \textcolor{dodgerblue}{R_{i}(1) = 1}, R_{i}(0) = 1] \Pr(R_{i}(1) = 1, \textcolor{dodgerblue}{R_{i}(0) = 1} \mid D_{i} = 1) \label{eq:mist_start}\\
    &\quad+ \E[Y_{i2}(0) - Y_{i1} \mid D_{i} = 0, \textcolor{dodgerblue}{R_{i}(1) = 1}, R_{i}(0) = 0] \Pr(R_{i}(1) = 1, \textcolor{dodgerblue}{R_{i}(0) = 0} \mid D_{i} = 1) \\
    &\quad+ \E[Y_{i2}(0) - Y_{i1} \mid D_{i} = 0, R_{i}(1) = 0, R_{i}(0) = 1] \Pr(R_{i}(1) = 0, R_{i}(0) = 1 \mid D_{i} = 1) \\
    &\quad+ \E[Y_{i2}(0) - Y_{i1} \mid D_{i} = 0, R_{i}(1) = 0, R_{i}(0) = 0] \Pr(R_{i}(1) = 0, R_{i}(0) = 0 \mid D_{i} = 1) \label{eq:mist_end}
\end{align}
We can further simplify this with two different strategies:
\begin{itemize}
    \item Assume missing independent of selection into treatment (See Remark~\ref{r:mist-ind}).
    \item Further extend the parallel trend assumption (e.g., parallel trend across principal strata; See Remark~\ref{r:mean-ind}).
\end{itemize}
\paragraph{Part 2} Using the law of total expectation and the principal strata parallel trends assumption,
\begin{align*}
    &\E[Y_{i2} - Y_{i1} \mid D_{i} = 1, R_{i} = 0] \Pr(R_{i} = 0 \mid D_{i} = 1) \\
    &=\E[\textcolor{indianred}{Y_{i2}}(1) - Y_{i2}(0) \mid D_{i} = 1, R_{i}(1) = 0, R_{i}(0) = 0] \Pr(R_{i}(1) = 0, R_{i}(0) = 0 \mid D_{i} = 1) \\
    &\quad+ \E[\textcolor{indianred}{Y_{i2}}(1) - Y_{i2}(0) \mid D_{i} = 1, R_{i}(1) = 0, R_{i}(0) = 1] \Pr(R_{i}(1) = 0, R_{i}(0) = 1 \mid D_{i} = 1) \\
    &\quad+ \E[Y_{i2}(0) - Y_{i1} \mid D_{i} = 0, R_{i}(1) = 0, R_{i}(0) = 1] \Pr(R_{i}(1) = 0, R_{i}(0) = 1 \mid D_{i} = 1) \\
    &\quad+ \E[Y_{i2}(0) - Y_{i1} \mid D_{i} = 0, R_{i}(1) = 0, R_{i}(0) = 0] \Pr(R_{i}(1) = 0, R_{i}(0) = 0 \mid D_{i} = 1)
\end{align*}
To further simplify this, it requires assumptions on the effect heterogeniety across principal strata.
\end{proof}

\subsection{Remarks on Section~\ref{sec:standard}} \label{app:standard-remark}

\begin{namedremark}{r:mist-ind}{Identification attempt under missing independent of selection into treatment}{
    If we assume $S_{i} \indep D_{i}$ (missing independent of selection into treatment), 
    the quantity can be further simplified as
    \begin{align*}
        &\E[Y_{i2}(1) - Y_{i2}(0) \mid D_{i} = 1]\\
        &= \E[Y_{i2} - Y_{i1} \mid D_{i} = 1, R_{i} = 1] \Pr(R_{i} = 1 \mid D_{i} = 1) \\
        &+ \E[\textcolor{indianred}{Y_{i2}} - Y_{i1} \mid D_{i} = 1, R_{i} = 0] \Pr(R_{i} = 0 \mid D_{i} = 1) \\
        &- \E[Y_{i2} - Y_{i1} \mid D_{i} = 0, R_{i} = 1] \Pr(R_{i} = 1 \mid D_{i} = 0) \\
        &- \E[\textcolor{indianred}{Y_{i2}} - Y_{i1} \mid D_{i} = 0, R_{i} = 0] \Pr(R_{i} = 0 \mid D_{i} = 0)
    \end{align*}
    where the terms in red are not identified due to the missingness. }
\end{namedremark}
\begin{proof}
    From Eq. \eqref{eq:ident_start} -- \eqref{eq:ident_end}, we have
    \begin{align*}
        &\E[Y_{i2}(1) - Y_{i2}(0) \mid D_{i} = 1]\\
        &= \E[Y_{i2} - Y_{i1} \mid D_{i} = 1, R_{i} = 1] \Pr(R_{i} = 1 \mid D_{i} = 1) \\
        &\quad+ \E[\textcolor{indianred}{Y_{i2}} - Y_{i1} \mid D_{i} = 1, R_{i} = 0] \Pr(R_{i} = 0 \mid D_{i} = 1) \\
        &\quad- \E[Y_{i2}(0) - Y_{i1} \mid D_{i} = 1, R_{i} = 1] \Pr(R_{i} = 1 \mid D_{i} = 1) \\
        &\quad- \E[Y_{i2}(0) - Y_{i1} \mid D_{i} = 1, R_{i} = 0] \Pr(R_{i} = 0 \mid D_{i} = 1)
    \end{align*}
    where the last two lines are equal to Eq. \eqref{eq:mist_start} -- \eqref{eq:mist_end} respectively.
    Substituting $\Pr(S_{i} = s \mid D_{i} = 1)$  with $\Pr(S_{i} = s \mid D_{i} = 0)$, we have
    \begin{align*}
        & \E[Y_{i2}(0) - Y_{i1} \mid D_{i} = 0, R_{i}(1) = 1, R_{i}(0) = 1] \Pr(R_{i}(1) = 1, R_{i}(0) = 1 \mid D_{i} = 0) \\
        &\quad+ \E[Y_{i2}(0) - Y_{i1} \mid D_{i} = 0, R_{i}(1) = 1, R_{i}(0) = 0] \Pr(R_{i}(1) = 1, R_{i}(0) = 0 \mid D_{i} = 0) \\
        &\quad+ \E[Y_{i2}(0) - Y_{i1} \mid D_{i} = 0, R_{i}(1) = 0, R_{i}(0) = 1] \Pr(R_{i}(1) = 0, R_{i}(0) = 1 \mid D_{i} = 0) \\
        &\quad+ \E[Y_{i2}(0) - Y_{i1} \mid D_{i} = 0, R_{i}(1) = 0, R_{i}(0) = 0] \Pr(R_{i}(1) = 0, R_{i}(0) = 0 \mid D_{i} = 0) \\
        &= \E[Y_{i2}(0) - Y_{i1} \mid D_{i} = 0, R_{i}(0) = 1] \Pr(R_{i}(0) = 1 \mid D_{i} = 0) \\
        &\quad+ \E[Y_{i2}(0) - Y_{i1} \mid D_{i} = 0, R_{i}(0) = 0] \Pr(R_{i}(0) = 0 \mid D_{i} = 0) \\
        &= \E[Y_{i2} - Y_{i1} \mid D_{i} = 0, R_{i} = 1] \Pr(R_{i} = 1 \mid D_{i} = 0) \\
        &\quad+ \E[\textcolor{indianred}{Y_{i2}} - Y_{i1} \mid D_{i} = 0, R_{i} = 0] \Pr(R_{i} = 0 \mid D_{i} = 0)
    \end{align*}
    Putting these together, we have
    \begin{align*}
        &\E[Y_{i2}(1) - Y_{i2}(0) \mid D_{i} = 1]\\
        &= \E[Y_{i2} - Y_{i1} \mid D_{i} = 1, R_{i} = 1] \Pr(R_{i} = 1 \mid D_{i} = 1) \\
        &\quad+ \E[\textcolor{indianred}{Y_{i2}} - Y_{i1} \mid D_{i} = 1, R_{i} = 0] \Pr(R_{i} = 0 \mid D_{i} = 1) \\
        &\quad- \E[Y_{i2} - Y_{i1} \mid D_{i} = 0, R_{i} = 1] \Pr(R_{i} = 1 \mid D_{i} = 0) \\
        &\quad- \E[\textcolor{indianred}{Y_{i2}} - Y_{i1} \mid D_{i} = 0, R_{i} = 0] \Pr(R_{i} = 0 \mid D_{i} = 0)
    \end{align*}
    This implies that missing independent of selection into treatment is not sufficient for identification.
\end{proof}

\begin{namedremark}{r:mean-ind}{Identification attempt under parallel trends across principal strata}{    
    Assume the following extended parallel trends assumption:
    \begin{equation*}
        \E[Y_{i2}(0) - Y_{i1} \mid D_{i} = d, S_{i} = s] = \E[Y_{i2}(0) - Y_{i1} \mid D_{i} = d^\prime, S_{i} = s^\prime]
    \end{equation*}
    for $d, d^\prime \in \{0,1\}$ and $s, s^\prime \in \{(0,0), (0,1), (1,0), (1,1)\}$.
    Somewhat trivial result is that, we can identify $\E[Y_{i2}(0) - Y_{i1} \mid D_{i} = 1]$ by $\E[Y_{i2} - Y_{i1} \mid D_{i} = 0, R_{i} = 1]$. However, we cannot identify ATT without further assumptions such as homogeneous effect across principal strata.}
\end{namedremark}
\begin{proof}
    \begin{align*}
        & \E[Y_{i2}(0) - Y_{i1} \mid D_{i} = 1] \\
        &= \sum_{r_1, r_0}\E[Y_{i2}(0) - Y_{i1} \mid D_{i} = 1, R_{i}(1) = r_1, R_{i}(0) = r_0] \Pr(R_{i}(1) = r_i, R_{i}(0) = r_0 \mid D_{i} = 1) \\
        &= \sum_{r_1, r_0}\E[Y_{i2}(0) - Y_{i1} \mid D_{i} = 0, R_{i}(1) = r_1, R_{i}(0) = r_0] \Pr(R_{i}(1) = r_i, R_{i}(0) = r_0 \mid D_{i} = 1) \\
        &= \E[Y_{i2}(0) - Y_{i1} \mid D_{i} = 0, R_{i}(1) = r_1^\prime, R_{i}(0) = r_0^\prime] \sum_{r_1, r_0} \Pr(R_{i}(1) = r_i, R_{i}(0) = r_0 \mid D_{i} = 1) \\
        &= \E[Y_{i2}(0) - Y_{i1} \mid D_{i} = 0, R_{i}(1) = 1, R_{i} = 1] \\
        &= \E[Y_{i2}(0) - Y_{i1} \mid D_{i} = 0, R_{i}(1) = 1, R_{i} = 1] \Pr(R_{i}(1) = 1 \mid R_{i} = 1, D_{i} = 0) \\
        &\quad+ \E[Y_{i2}(0) - Y_{i1} \mid D_{i} = 0, R_{i}(1) = 1, R_{i} = 1] \Pr(R_{i}(1) = 0 \mid R_{i} = 1, D_{i} = 0) \\
        &= \E[Y_{i2}(0) - Y_{i1} \mid D_{i} = 0, R_{i}(1) = 1, R_{i} = 1] \Pr(R_{i}(1) = 1 \mid R_{i} = 1, D_{i} = 0) \\
        &\quad+ \E[Y_{i2}(0) - Y_{i1} \mid D_{i} = 0, R_{i}(1) = 0, R_{i} = 1] \Pr(R_{i}(1) = 0 \mid R_{i} = 1, D_{i} = 0) \\
        &= \E[Y_{i2}(0) - Y_{i1} \mid D_{i} = 0, R_{i} = 1] \\
        &= \E[Y_{i2} - Y_{i1} \mid D_{i} = 0, R_{i} = 1]
    \end{align*}
\end{proof}

\newpage
\section{Proofs of Section~\ref{sec:proposed}} \label{app:proposed}

\subsection{Proofs of Section~\ref{sec:proposed-att}} \label{app:proposed-att}

\paragraph{Setup.}
Suppose we have multiple waves of panel data where there exists missing data in the pre-treatment periods. From now on, we use $R_{it}$ instead of $R_{i}$ to denote the response indicator of the outcome variable at time $t$. That is, $R_{it} = 1$ if $Y_{it}$ is observed, and $0$ if it is missing, for $t = 1,2$.
Furthermore, we introduce an auxiliary variable, $W_{i}$ and its response indicators, $\widetilde{R}_{i}$.
The auxiliary variable can be constructed from other pre-treatment variables.
For instance, in the first motivating example, we may use other outcome variables (e.g. ``Afghanistan right direction?'', etc) at $2016$ wave as auxiliary variables. 

\begin{table}[ht!] 
	\centering
        \begin{tabular}{c|cccc} \toprule
            ID & $W_{i}$ & $Y_{i1}$ & $D_i$ & $Y_{i2}$ \\ \midrule
            1 & 0 & 0 & 0 & 1 \\
            2 & 1 & \texttt{NA} & 0 & 2 \\
            3 & 1 & \texttt{NA} & 1 & 2 \\
            4 & \texttt{NA} & 2 & 0 & 3 \\
            \rowcolor{gray!30} 5 & 1 & 2 & 1 & 4 \\
            6 & 2 & 3 & 1 & \texttt{NA} \\
            \rowcolor{gray!30} 7 & \texttt{NA} & 0 & 1 & 1 \\
            8 & \texttt{NA} & 1 & 1 & \texttt{NA} \\
            9 & \texttt{NA} & 2 & 1 & \texttt{NA} \\
            \bottomrule
        \end{tabular}
        \hspace{2em}
        \begin{tabular}{c|cccc} \toprule
            ID & $\widetilde{R}_{i}$ & $R_{i1}$ & $D_i$ & $R_{i2}$ \\ \midrule
            1 & 1 & 1 & 0 & 1 \\
            2 & 1 & 0 & 0 & 1 \\
            3 & 1 & 0 & 1 & 1 \\
            4 & 0 & 1 & 0 & 1 \\
            \rowcolor{yellow!30} 5 & 1 & 1 & 1 & 1 \\
            \rowcolor{yellow!10} 6 & 1 & 1 & 1 & 0 \\
            \rowcolor{yellow!30} 7 & 0 & 1 & 1 & 1 \\
            \rowcolor{yellow!10} 8 & 0 & 1 & 1 & 0 \\
            \rowcolor{yellow!10} 9 & 0 & 1 & 1 & 0 \\
            \bottomrule
        \end{tabular}
	\caption{Toy Example of Observed Data with Three Waves.}
    \label{tab:toy}
\end{table}

Note that we now allow for the missingness in the pre-treatment periods, which may make DID estimator infeasible unless we assume further about the missingness mechanism of the pre-treatment outcome. 
Here, we assume an exclusion restriction type of assumption for the pre-treatment outcome missingness and time trend of outcome.

\begin{namedassumption}{a:pt_mar}{Exclusion restriction for pre-treatment outcome missingness}{
    The pre-treatment outcome missingness may affect the time trend of the outcome only through the treatment selection and the post-treatment outcome missingness.
    \begin{align*}
        Y_{i2} - Y_{i1} &\indep R_{i1} \mid D_{i}, R_{i2} 
    \end{align*}
}
\end{namedassumption}

One may consider this assumption as a variant of the missing at random (MAR) assumption for the pre-treatment outcome missingness, where the missingness of the pre-treatment outcome is independent of the missing values of the time trend of the outcome given other information including the treatment selection and the post-treatment outcome missingness. 
To be precise, this is not to say that the pre-treatment outcome is MAR, since the assumption is not about the missing values of $Y_{i1}$ itself but about the missingness of the time trend of the outcome $Y_{i2} - Y_{i1}$.
Based on this assumption, we introduce a set of assumptions for using the baseline response indicator as IV for the time trend of the outcome, conditioning on the treatment and the pre-treatment outcome missingness. 
Note that a weaker version of the assumption is required for the identification ($\E[Y_{i2} - Y_{i1} \mid D_{i} = d, R_{i2} = r, R_{i1} = 1] = \E[Y_{i2} - Y_{i1} \mid D_{i} = d, R_{i2} = r, R_{i1} = 0]$), yet we consider the stronger version for the sake of clarity.
Alternatively, we can also define ATT for the population who responded to the pre-treatment survey, i.e. $R_{i1} = 1$.

\begin{lemma} \label{lem:inst}
    Under \Cref{a:missing_iv} and \Cref{a:pt_mar}, we have
    \begin{align*}
        &\E[Y_{i2} - Y_{i1} \mid D_{i} = d, R_{i2} = 0] \\
        &=\E[Y_{i2} - Y_{i1} \mid D_{i} = d, R_{i1} = 1, R_{i2} = 1]\\
        &\quad+ \big\{ \E[Y_{i2} - Y_{i1} \mid D_{i} = d, \widetilde{R}_{i} = 1, R_{i1} = 1, R_{i2} = 1] - \E[Y_{i2} - Y_{i1} \mid D_{i} = d, \widetilde{R}_{i} = 0, R_{i1} = 1, R_{i2} = 1]\big\}\\
        &\qquad\times \big\{\Pr(R_{i2} = 0 \mid D_{i} = d, \widetilde{R}_{i} = 0, R_{i1} = 1) - \Pr(R_{i2} = 0 \mid D_{i} = d, \widetilde{R}_{i} = 1, R_{i1} = 1)\big\}^{-1}
    \end{align*}
    for $d = 0,1$.
\end{lemma}
\begin{proof}
    \begin{align*}
        &\E[Y_{i2} - Y_{i1} \mid D_{i} = 1, \widetilde{R}_{i} = 1, R_{i1} = 1] \\
        &=\sum_{r= 0,1}\E[Y_{i2} - Y_{i1} \mid D_{i} = 1, \widetilde{R}_{i} = 1, R_{i1} = 1, R_{i2} = r] \Pr(R_{i2} = r \mid D_{i} = 1, \widetilde{R}_{i} = 1, R_{i1}  = 1) \\
        &=\sum_{r= 0,1}\E[Y_{i2} - Y_{i1} \mid D_{i} = 1, \widetilde{R}_{i} = 1, R_{i1} = 1, R_{i2} = r] \Pr(R_{i2} = r \mid D_{i} = 1, \widetilde{R}_{i} = 1, R_{i1}  = 1) \\
        &\qquad + 
        \E[Y_{i2} - Y_{i1} \mid D_{i} = 1, \widetilde{R}_{i} = 1, R_{i1} = 1, R_{i2} = 1]
        \Pr(R_{i2} = 0 \mid D_{i} = 1, \widetilde{R}_{i} = 1, R_{i1}  = 1)\\
        &\qquad -
        \E[Y_{i2} - Y_{i1} \mid D_{i} = 1, \widetilde{R}_{i} = 1, R_{i1} = 1, R_{i2} = 1]
        \Pr(R_{i2} = 0 \mid D_{i} = 1, \widetilde{R}_{i} = 1, R_{i1}  = 1)\\
        &=\E[Y_{i2} - Y_{i1} \mid D_{i} = d, \widetilde{R}_{i} = 1, R_{i1} = 1, R_{i2} = 1]\\ 
        &\qquad+\{\E[Y_{i2} - Y_{i1} \mid D_{i} = 1, \widetilde{R}_{i} = 1, R_{i1} = 1, R_{i2} = 0] - \E[Y_{i2} - Y_{i1} \mid D_{i} = 1, \widetilde{R}_{i} = 1, R_{i1} = 1, R_{i2} = 1]\}\\
        &\qquad\quad\times\Pr(R_{i2} = 0 \mid D_{i} = 1, \widetilde{R}_{i} = 1, R_{i1}  = 1)
    \end{align*}
    With similar steps, 
    \begin{align*}
        &\E[Y_{i2} - Y_{i1} \mid D_{i} = 1, \widetilde{R}_{i} = 0, R_{i1} = 1] \\
        &=\E[Y_{i2} - Y_{i1} \mid D_{i} = 1, \widetilde{R}_{i} = 0, R_{i1} = 1, R_{i2} = 1]\\ 
        &\qquad+\{\E[Y_{i2} - Y_{i1} \mid D_{i} = 1, \widetilde{R}_{i} = 0, R_{i1} = 1, R_{i2} = 0] - \E[Y_{i2} - Y_{i1} \mid D_{i} = 1, \widetilde{R}_{i} = 0, R_{i1} = 1, R_{i2} = 1]\}\\
        &\qquad\quad\times\Pr(R_{i2} = 0 \mid D_{i} = 1, \widetilde{R}_{i} = 0, R_{i1}  = 1) \\
        &=\E[Y_{i2} - Y_{i1} \mid D_{i} = 1, \widetilde{R}_{i} = 0, R_{i1} = 1, R_{i2} = 1]\\ 
        &\qquad+\{\E[Y_{i2} - Y_{i1} \mid D_{i} = 1, \widetilde{R}_{i} = 1, R_{i1} = 1, R_{i2} = 0] - \E[Y_{i2} - Y_{i1} \mid D_{i} = 1, \widetilde{R}_{i} = 1, R_{i1} = 1, R_{i2} = 1]\}\\
        &\qquad\quad\times\Pr(R_{i2} = 0 \mid D_{i} = 1, \widetilde{R}_{i} = 0, R_{i1}  = 1)
    \end{align*}
    By \textbf{parallel trends of observed outcome}, we have
    \begin{align*}
        \E[Y_{i2} - Y_{i1} \mid D_{i} = 1, \widetilde{R}_{i} = 1, R_{i1} = 1] =\E[Y_{i2} - Y_{i1} \mid D_{i} = 1, \widetilde{R}_{i} = 0, R_{i1} = 1]
    \end{align*}
    Thus, 
    \begin{align*}
        &\E[Y_{i2} - Y_{i1} \mid D_{i} = 1, \widetilde{R}_{i} = 1, R_{i1} = 1, R_{i2} = 1]\\ 
        &\qquad+\{\E[Y_{i2} - Y_{i1} \mid D_{i} = 1, \widetilde{R}_{i} = 1, R_{i1} = 1, R_{i2} = 0] - \E[Y_{i2} - Y_{i1} \mid D_{i} = 1, \widetilde{R}_{i} = 1, R_{i1} = 1, R_{i2} = 1]\}\\
        &\qquad\quad\times\Pr(R_{i2} = 0 \mid D_{i} = 1, \widetilde{R}_{i} = 1, R_{i1}  = 1)\\
        &=\E[Y_{i2} - Y_{i1} \mid D_{i} = 1, \widetilde{R}_{i} = 0, R_{i1} = 1, R_{i2} = 1]\\ 
        &\qquad+\{\E[Y_{i2} - Y_{i1} \mid D_{i} = 1, \widetilde{R}_{i} = 1, R_{i1} = 1, R_{i2} = 0] - \E[Y_{i2} - Y_{i1} \mid D_{i} = 1, \widetilde{R}_{i} = 1, R_{i1} = 1, R_{i2} = 1]\}\\
        &\qquad\quad\times\Pr(R_{i2} = 0 \mid D_{i} = 1, \widetilde{R}_{i} = 0, R_{i1}  = 1)
    \end{align*}
    which leads to
    \begin{align*}
        &\E[Y_{i2} - Y_{i1} \mid D_{i} = 1, \widetilde{R}_{i} = 1, R_{i1} = 1, R_{i2} = 0]\\
        &=\E[Y_{i2} - Y_{i1} \mid D_{i} = 1, \widetilde{R}_{i} = 1, R_{i1} = 1, R_{i2} = 1] \\
        &\quad+ \big\{ \E[Y_{i2} - Y_{i1} \mid D_{i} = 1, \widetilde{R}_{i} = 1, R_{i1} = 1, R_{i2} = 1] - \E[Y_{i2} - Y_{i1} \mid D_{i} = 1, \widetilde{R}_{i} = 0, R_{i1} = 1, R_{i2} = 1]\big\}\\
        &\qquad\times \big\{\Pr(R_{i2} = 0 \mid D_{i} = 1, \widetilde{R}_{i} = 0, R_{i1} = 1) - \Pr(R_{i2} = 0 \mid D_{i} = 1, \widetilde{R}_{i} = 1, R_{i1} = 1)\big\}^{-1}
    \end{align*}
    and 
    \begin{align*}
        &\E[Y_{i2} - Y_{i1} \mid D_{i} = 1, \widetilde{R}_{i} = 0, R_{i1} = 1, R_{i2} = 0]\\
        &=\E[Y_{i2} - Y_{i1} \mid D_{i} = 1, \widetilde{R}_{i} = 0, R_{i1} = 1, R_{i2} = 1] \\
        &\quad+ \big\{ \E[Y_{i2} - Y_{i1} \mid D_{i} = 1, \widetilde{R}_{i} = 1, R_{i1} = 1, R_{i2} = 1] - \E[Y_{i2} - Y_{i1} \mid D_{i} = 1, \widetilde{R}_{i} = 0, R_{i1} = 1, R_{i2} = 1]\big\}\\
        &\qquad\times \big\{\Pr(R_{i2} = 0 \mid D_{i} = 1, \widetilde{R}_{i} = 0, R_{i1} = 1) - \Pr(R_{i2} = 0 \mid D_{i} = 1, \widetilde{R}_{i} = 1, R_{i1} = 1)\big\}^{-1}
    \end{align*}
    By \textbf{bias homogeneity},
    \begin{align*}
        &\E[Y_{i2} - Y_{i1} \mid D_{i} = 1, R_{i1} = 1, R_{i2} = 1] - \E[Y_{i2} - Y_{i1} \mid D_{i} = 1, R_{i1} = 1, R_{i2} = 0] \\
        &=\E[Y_{i2} - Y_{i1} \mid D_{i} = 1, \widetilde{R}_{i} = 1, R_{i1} = 1, R_{i2} = 1] - \E[Y_{i2} - Y_{i1} \mid D_{i} = 1, \widetilde{R}_{i} = 1, R_{i1} = 1, R_{i2} = 0] \\
        &=\E[Y_{i2} - Y_{i1} \mid D_{i} = 1, \widetilde{R}_{i} = 0, R_{i1} = 1, R_{i2} = 1] - \E[Y_{i2} - Y_{i1} \mid D_{i} = 1, \widetilde{R}_{i} = 0, R_{i1} = 1, R_{i2} = 0] \\
        &= -\big\{ \E[Y_{i2} - Y_{i1} \mid D_{i} = 1, \widetilde{R}_{i} = 1, R_{i1} = 1, R_{i2} = 1] - \E[Y_{i2} - Y_{i1} \mid D_{i} = 1, \widetilde{R}_{i} = 0, R_{i1} = 1, R_{i2} = 1]\big\}\\
        &\qquad\times \big\{\Pr(R_{i2} = 0 \mid D_{i} = 1, \widetilde{R}_{i} = 0, R_{i1} = 1) - \Pr(R_{i2} = 0 \mid D_{i} = 1, \widetilde{R}_{i} = 1, R_{i1} = 1)\big\}^{-1}
    \end{align*}
    Plugging this into the following equation,
    \begin{align*}
        &\E[Y_{i2} - Y_{i1} \mid D_{i} = 1, R_{i2} = 0] \\
        &=\E[Y_{i2} - Y_{i1} \mid D_{i} = 1, R_{i1} = 1, R_{i2} = 0] \\
        &=\E[Y_{i2} - Y_{i1} \mid D_{i} = 1, R_{i1} = 1, R_{i2} = 1]\\
        &\quad+ \big\{ \E[Y_{i2} - Y_{i1} \mid D_{i} = 1, \widetilde{R}_{i} = 1, R_{i1} = 1, R_{i2} = 1] - \E[Y_{i2} - Y_{i1} \mid D_{i} = 1, \widetilde{R}_{i} = 0, R_{i1} = 1, R_{i2} = 1]\big\}\\
        &\qquad\times \big\{\Pr(R_{i2} = 0 \mid D_{i} = 1, \widetilde{R}_{i} = 0, R_{i1} = 1) - \Pr(R_{i2} = 0 \mid D_{i} = 1, \widetilde{R}_{i} = 1, R_{i1} = 1)\big\}^{-1}
    \end{align*}
    Similarly, we can show the equation with $d = 0$.
\end{proof}

\subsection{Proofs of Section~\ref{sec:proposed-partial}} \label{app:proposed-partial}
\subsubsection{Proof of Proposition~\ref{prop:bound}}

\begin{proof}
    By \Cref{a:pt-missing}, we have
    \begin{equation*}
        \E[R_{i2}(0)\mid D_{i} = 1] = \E[R_{i2}(0) - R_{i1}(0) \mid D_{i} = 0] + \E[R_{i1}(0) \mid D_{i} = 1].
    \end{equation*}
    Since $R_{it}$ is binary, we have
    \begin{equation*}
        \Pr(R_{i2}(0) = 1 \mid D_{i} = 1) = \Pr(R_{i2} = 1 \mid D_{i} = 0) - \Pr(R_{i1} = 1 \mid D_{i} = 0) + \Pr(R_{i1} = 1 \mid D_{i} = 1).
    \end{equation*}
    By \Cref{a:homogeneous}, we have
    \begin{equation*}
        \E[R_{i2} (1) \mid D_{i} = 0] = \E[R_{i2} (1) - R_{i2} (0) \mid D_{i} = 1] + \E[R_{i2} (0) \mid D_{i} = 0]
    \end{equation*}
    which gives us
    \begin{equation*}
        \Pr(R_{i2} (1) = 1 \mid D_{i} = 0) = \Pr(R_{i2} = 1 \mid D_{i} = 1) - \Pr(R_{i2} (0) = 1 \mid D_{i} = 1) + \Pr(R_{i2} = 1 \mid D_{i} = 0).
    \end{equation*}
    By plugging in the previous result, we have
    \begin{equation*}
        \Pr(R_{i2} (1) = 1 \mid D_{i} = 0) = \Pr(R_{i2} = 1 \mid D_{i} = 1) - \Pr(R_{i1} = 1 \mid D_{i} = 1) + \Pr(R_{i1} = 1 \mid D_{i} = 0).
    \end{equation*}
\end{proof}

\subsubsection{Additional Result: Bounds for Principal Strata Proportion}
Without monotonicity, we can bound the proportion of principal strata, $\pi_{11}(0)$ and $\pi_{11}(1)$, using \Cref{a:pt-missing} and \ref{a:homogeneous}.
\begin{align*}
    &\max\left\{0, \Pr(R_{i2} (0) = 1 \mid D_i = 1) - \Pr(R_{i2}  = 0 \mid D_i = 1)\right\} \\
    &\quad\le \pi_{11}(1) \le \min\left\{\Pr(R_{i2}  = 1 \mid D_i = 1), \Pr(R_{i2} (0) = 1 \mid D_i = 1)\right\} \\
    &\max\left\{0, \Pr(R_{i2} (1) = 1 \mid D_i = 0) - \Pr(R_{i2}  = 0 \mid D_i = 0)\right\} \\
    &\quad\le \pi_{11}(0) \le \min\left\{\Pr(R_{i2}  = 1 \mid D_i = 0), \Pr(R_{i2} (1) = 1 \mid D_i = 0)\right\}
\end{align*}
where
\begin{align*}
    \Pr(R_{i2} (0) = 1 \mid D_i = 1) &= \Pr(R_{i2} = 1 \mid D_{i} = 0) - \Pr(R_{i1} = 1 \mid D_{i} = 0) + \Pr(R_{i1}  = 1 \mid D_{i} = 1) \\
    \Pr(R_{i2} (1) = 1 \mid D_i = 0) &= \Pr(R_{i2} = 1 \mid D_{i} = 1) - \Pr(R_{i1} = 1 \mid D_{i} = 1) + \Pr(R_{i1}  = 1 \mid D_{i} = 0).
\end{align*}

\begin{proof}
    Part 1: Bounding $\pi_{11}(1)$. We have the following observed quantities and constraint
    \begin{equation*}
        \begin{cases}
            \pi_{r_1, r_0}(1) \in [0,1] \quad \forall r_0, r_1 = 0,1 \\
            \pi_{11}(1) + \pi_{01}(1) + \pi_{10}(1) + \pi_{00}(1) = 1 \\
            \pi_{10}(1) = \Pr(R_{i2}  = 1 \mid D_i = 1) - \pi_{11}(1) \\
            \pi_{01}(1) = \Pr(R_{i2} (0) = 1 \mid D_i = 1) - \pi_{11}(1) 
        \end{cases}
    \end{equation*}
    where $\Pr(R_{i2} (0) = 1 \mid D_i = 1)$ is identified with Assumption~\ref{a:pt-missing}.
    Then, combining last three conditions we have
    \begin{align*}
        \pi_{00}(1) &= 1 - \Pr(R_{i2}  = 1 \mid D_i = 1) - \Pr(R_{i2} (0) = 1 \mid D_i = 1) + \pi_{11}(1) \\
        &= \Pr(R_{i2}  = 0 \mid D_i = 1) - \Pr(R_{i2} (0) = 1 \mid D_i = 1) + \pi_{11}(1)
    \end{align*}
    Thus, we have
    \begin{equation*}
        \begin{cases}
            0 \le \pi_{11}(1) \le 1 \\
            -\Pr(R_{i2}  = 0 \mid D_i = 1)\le  \pi_{11}(1) \le \Pr(R_{i2}  = 1 \mid D_i = 1) \\
            -\Pr(R_{i2} (0) = 0 \mid D_i = 1) \le  \pi_{11}(1) \le \Pr(R_{i2} (0) = 1 \mid D_i = 1) \\
            \Pr(R_{i2} (0) = 1 \mid D_i = 1) - \Pr(R_{i2}  = 0 \mid D_i = 1)\le \pi_{11}(1) \le \Pr(R_{i2}  = 1 \mid D_i = 1) + \Pr(R_{i2} (0) = 1 \mid D_i = 1)
        \end{cases}
    \end{equation*}
    which gives us the bounds for $\pi_{11}(1)$
    \begin{align*}
        &\max\left\{0, \Pr(R_{i2} (0) = 1 \mid D_i = 1) - \Pr(R_{i2}  = 0 \mid D_i = 1)\right\} \\
        &\quad\le \pi_{11}(1) \le \min\left\{\Pr(R_{i2}  = 1 \mid D_i = 1), \Pr(R_{i2} (0) = 1 \mid D_i = 1)\right\}
    \end{align*}

    Part 2: Bounding $\pi_{11}(0)$. By the assumption $\E[R_{i}(1) - R_{i}(0) \mid D_{i} = 1] = \E[R_{i}(1) - R_{i}(0) \mid D_{i} = 0]$ we have $\Pr(R_{i}(1) = 1 \mid D_{i} = 0) = \Pr(R_{i} = 1 \mid D_{i} = 1) + \Pr(R_{i} = 1 \mid D_{i} = 0) - \Pr(R_{i}(0) = 1 \mid D_{i} = 1)$. Plugging in $\Pr(R_{i2} (0) = 1 \mid D_i = 1) = \Pr(R_{i} = 1 \mid D_{i} = 0) - \Pr(R_{i1}  = 1 \mid D_{i} = 0) + \Pr(R_{i1}  = 1 \mid D_{i} = 1)$, we have
    \begin{equation*}
        \Pr(R_{i}(1) = 1 \mid D_{i} = 0) = \Pr(R_{i} = 1 \mid D_{i} = 1) - \Pr(R_{i1}  = 1 \mid D_{i} = 1) + \Pr(R_{i1}  = 1 \mid D_{i} = 0)
    \end{equation*}
    Accordingly, we have
    \begin{equation*}
        \begin{cases}
            \pi_{r_1, r_0}(0) \in [0,1] \quad \forall r_0, r_1 = 0,1 \\
            \pi_{11}(0) + \pi_{01}(0) + \pi_{10}(0) + \pi_{00}(0) = 1 \\
            \pi_{10}(0) = \Pr(R_{i2}  = 1 \mid D_i = 0) - \pi_{11}(0) \\
            \pi_{01}(0) = \Pr(R_{i2} (1) = 1 \mid D_i = 0) - \pi_{11}(0) 
        \end{cases}
    \end{equation*}
    Analogous to the previous case, we have
    \begin{align*}
        &\max\left\{0, \Pr(R_{i2} (1) = 1 \mid D_i = 0) - \Pr(R_{i2}  = 0 \mid D_i = 0)\right\} \\
        &\quad\le \pi_{11}(0) \le \min\left\{\Pr(R_{i2}  = 1 \mid D_i = 0), \Pr(R_{i2} (1) = 1 \mid D_i = 0)\right\}
    \end{align*}
\end{proof}

\newpage
\section{Extension to Multiple Baseline Response Indicators} \label{app:multi}

\begin{table}[ht!]
    \centering
        \begin{tabular}{c|ccccc} \toprule
            ID & $W^{(1)}_{i}$  & $W^{(2)}_{i}$ & $Y_{i1}$ & $D_i$ & $Y_{i2}$ \\ \midrule
            1 & 0  & 0 & 0 & 0 & 1 \\
            2 & 2  & 1 & \texttt{NA} & 0 & \texttt{NA} \\
            3 & 0 & 1 & 3 & 1 & \texttt{NA} \\
            4 & 3 & \texttt{NA} & \texttt{NA} & 0 & \texttt{NA} \\
            \cellcolor{purple!20}5 &\cellcolor{purple!20}1 & 1 &\cellcolor{purple!20}2 &\cellcolor{purple!20}1 &\cellcolor{purple!20}4 \\
            6 & 2 & 2 & 3 & 1 & \texttt{NA} \\
            \cellcolor{purple!20}7 &\cellcolor{purple!20}2 & 1 &\cellcolor{purple!20}0 &\cellcolor{purple!20}1 &\cellcolor{purple!20}1 \\
            8 & 1 & \texttt{NA} & 3 & 1 & \texttt{NA} \\
            9 & \texttt{NA} & 1 & 2 & 1 & 4 \\
            \bottomrule
        \end{tabular}
        \hspace{2em}
        \begin{tabular}{c|ccccc} \toprule
            ID & $\widetilde{R}^{(1)}_{i}$& $\widetilde{R}^{(2)}_{i}$ & ${R}_{i1}$ & $D_i$ & $R_{i2}$ \\ \midrule
            1 & 1  & 1 & 1 & 0 & 1 \\
            2 & 1  & 1 & 0 & 0 & 0 \\
            3 & 1 & 1 & 1 & 1 & 0 \\
            4 & 1 & 0 & 0 & 0 & 0 \\
            \cellcolor{yellow!30}5 &\cellcolor{yellow!30}1 & 1 &\cellcolor{yellow!30}1 &\cellcolor{yellow!30}1 &\cellcolor{yellow!30}1 \\
            \cellcolor{yellow!10}6 &\cellcolor{yellow!10}1 & 1 &\cellcolor{yellow!10}1 &\cellcolor{yellow!10}1 &\cellcolor{yellow!10}0 \\
            \cellcolor{yellow!30}7 &\cellcolor{yellow!30}1 & 1 &\cellcolor{yellow!30}1 &\cellcolor{yellow!30}1 &\cellcolor{yellow!30}1 \\
            \cellcolor{yellow!10}8 &\cellcolor{yellow!10}1 & 0 &\cellcolor{yellow!10}1 &\cellcolor{yellow!10}1 &\cellcolor{yellow!10}0 \\
            9 & 0 & 1 & 1 & 1 & 1 \\
            \bottomrule
        \end{tabular}
        \begin{tabular}{c|ccccc} \toprule
            ID & $W^{(1)}_{i}$  & $W^{(2)}_{i}$ & $Y_{i1}$ & $D_i$ & $Y_{i2}$ \\ \midrule
            1 & 0  & 0 & 0 & 0 & 1 \\
            2 & 2  & 1 & \texttt{NA} & 0 & \texttt{NA} \\
            3 & 0 & 1 & 3 & 1 & \texttt{NA} \\
            4 & 3 & \texttt{NA} & \texttt{NA} & 0 & \texttt{NA} \\
            \cellcolor{purple!20}5 & 1 &\cellcolor{purple!20}1 &\cellcolor{purple!20}2 &\cellcolor{purple!20}1 &\cellcolor{purple!20}4 \\
            6 & 2 & 2 & 3 & 1 & \texttt{NA} \\
            \cellcolor{purple!20}7 & 2 &\cellcolor{purple!20}1 &\cellcolor{purple!20}0 &\cellcolor{purple!20}1 &\cellcolor{purple!20}1 \\
            8 & 1 & \texttt{NA} & 3 & 1 & \texttt{NA} \\
            \cellcolor{purple!20}9 & \texttt{NA} &\cellcolor{purple!20}1 &\cellcolor{purple!20} 2 &\cellcolor{purple!20}1 &\cellcolor{purple!20}4 \\
            \bottomrule
        \end{tabular}
        \hspace{2em}
        \begin{tabular}{c|ccccc} \toprule
            ID & $\widetilde{R}^{(1)}_{i}$& $\widetilde{R}^{(2)}_{i}$ & ${R}_{i1}$ & $D_i$ & $R_{i2}$ \\ \midrule
            1 & 1  & 1 & 1 & 0 & 1 \\
            2 & 1  & 1 & 0 & 0 & 0 \\
            3 & 1 & 1 & 1 & 1 & 0 \\
            4 & 1 & 0 & 0 & 0 & 0 \\
            \cellcolor{yellow!30}5 & 1 &\cellcolor{yellow!30}1 &\cellcolor{yellow!30}1 &\cellcolor{yellow!30}1 &\cellcolor{yellow!30}1 \\
            \cellcolor{yellow!10}6 & 1 &\cellcolor{yellow!10}1 &\cellcolor{yellow!10}1 &\cellcolor{yellow!10}1 &\cellcolor{yellow!10}0 \\
            \cellcolor{yellow!30}7 & 1 &\cellcolor{yellow!30}1 &\cellcolor{yellow!30}1 &\cellcolor{yellow!30}1 &\cellcolor{yellow!30}1 \\
            8 & 1 & 0 & 1 & 1 & 0 \\
            \cellcolor{yellow!30}9 & 0 &\cellcolor{yellow!30}1 &\cellcolor{yellow!30}1 &\cellcolor{yellow!30}1 &\cellcolor{yellow!30}1 \\
            \bottomrule
        \end{tabular}
    \caption{Toy Example of Observed Data with Auxiliary Variables.}
\end{table}

\begin{namedassumption}{a:missing_multiple_iv}{Baseline response indicators as IV for time trend}{
    Suppose we have baseline response indicators $\widetilde{R}^{(1)}_{i}$ and $\widetilde{R}^{(2)}_{i}$. We assume the following set of assumptions, for each treatment group whose pre-treatment outcome is observed: 
    \begin{enumerate}
        \item \textbf{Parallel difference in trends} The difference in time trends of the outcome between respondents and nonrespondents is parallel for two auxiliary variables.
        \begin{align*}
            &\E[Y_{i2} - Y_{i1} \mid D_i = d, \widetilde{R}^{(1)}_{i} = 1, R_{i1} = 1] - \E[Y_{i2} - Y_{i1} \mid D_i = d, \widetilde{R}^{(1)}_{i} = 0, R_{i1} = 1] \\
            &= \E[Y_{i2} - Y_{i1} \mid D_i = d, \widetilde{R}^{(2)}_{i} = 1, R_{i1} = 1] - \E[Y_{i2} - Y_{i1} \mid D_i = d, \widetilde{R}^{(2)}_{i} = 0, R_{i1} = 1]
        \end{align*}
        for $d = 0,1$.
        \item \textbf{Bias homogeneity}: The bias due to missingness in the time trend of the outcome is homogeneous across the subgroups defined by the baseline response indicators, and is same as the marginalized version.
        \begin{align*}
            &\E[Y_{i2} - Y_{i1} \mid D_i = d, R_{i1} = 1, R_{i2} = 1] - \E[Y_{i2} - Y_{i1} \mid D_i = d, R_{i1} = 1, R_{i2} = 0]\\
            &= \E[Y_{i2} - Y_{i1} \mid D_i = d, \widetilde{R}^{(1)}_{i} = r, R_{i1} = 1, R_{i2} = 1] - \E[Y_{i2} - Y_{i1} \mid D_i = d, \widetilde{R}^{(1)}_{i} = r, R_{i1} = 1, R_{i2} = 0] \\
            &=\E[Y_{i2} - Y_{i1} \mid D_i = d, \widetilde{R}^{(2)}_{i} = r, R_{i1} = 1, R_{i2} = 1] - \E[Y_{i2} - Y_{i1} \mid D_i = d, \widetilde{R}^{(2)}_{i} = r, R_{i1} = 1, R_{i2} = 0]
        \end{align*}
        for $d = 0,1$ and $r = 0,1$.
        \item \textbf{Relevance to missingness}: The baseline response indicators are relevant to the missingness of post-treatment outcome.
        \begin{align*}
            \Pr(R_{i2} = 0 \mid D_{i} = d, \widetilde{R}^{(1)}_{i} = 0, R_{i1} = 1) &\neq \Pr(R_{i2} = 0 \mid D_{i} = d, \widetilde{R}^{(1)}_{i} = 1, R_{i1} = 1) \\
            \Pr(R_{i2} = 0 \mid D_{i} = d, \widetilde{R}^{(2)}_{i} = 0, R_{i1} = 1) &\neq \Pr(R_{i2} = 0 \mid D_{i} = d, \widetilde{R}^{(2)}_{i} = 1, R_{i1} = 1)
        \end{align*}
        for $d = 0,1$.
    \end{enumerate}
}
\end{namedassumption}

\begin{lemma} \label{lem:inst_multiple}
    Under Assumption~\ref{a:missing_multiple_iv}, we have
    \begin{align*}
        &\E[Y_{i2} - Y_{i1} \mid D_{i} = d,  \widetilde{R}^{(1)}_{i} = 1, R_{i1} = 1, R_{i2} = 0]-\E[Y_{i2} - Y_{i1} \mid D_{i} = d,  \widetilde{R}^{(1)}_{i} = 1, R_{i1} = 1, R_{i2} = 1] \\
        &= \big[\E[Y_{i2} - Y_{i1} \mid D_{i} = d,  \widetilde{R}^{(1)}_{i} = 1, R_{i1} = 1, R_{i2} = 1]-\E[Y_{i2} - Y_{i1} \mid D_{i} = d,  \widetilde{R}^{(1)}_{i} = 0, R_{i1} = 1, R_{i2} = 1]\\
        &\qquad- \{\E[Y_{i2} - Y_{i1} \mid D_{i} = d,  \widetilde{R}^{(2)}_{i} = 1, R_{i1} = 1, R_{i2} = 1]-\E[Y_{i2} - Y_{i1} \mid D_{i} = d,  \widetilde{R}^{(2)}_{i} = 0, R_{i1} = 1, R_{i2} = 1]\}\big]\\
        &\quad\times \big[\Pr(R_{i2} = 0 \mid D_{i} = d,  \widetilde{R}^{(1)}_{i} = 0, R_{i1} = 1) - \Pr(R_{i2} = 0 \mid D_{i} = d,  \widetilde{R}^{(1)}_{i} = 1, R_{i1} = 1)\\
        &\qquad- \{\Pr(R_{i2} = 0 \mid D_{i} = d,  \widetilde{R}^{(2)}_{i} = 1, R_{i1} = 1) - \Pr(R_{i2} = 0 \mid D_{i} = d,  \widetilde{R}^{(2)}_{i} = 0, R_{i1} = 1)\} \big]^{-1}
    \end{align*}
    for $d = 0,1$.
\end{lemma}
\begin{proof}
    \begin{align*}
        &\E[Y_{i2} - Y_{i1} \mid D_{i} = 1, \widetilde{R}^{(1)}_{i} = 1, R_{i1} = 1] \\
        &=\sum_{r= 0,1}\E[Y_{i2} - Y_{i1} \mid D_{i} = 1,  \widetilde{R}^{(1)}_{i} = 1, R_{i1} = 1, R_{i2} = r] \Pr(R_{i2} = r \mid D_{i} = 1,  \widetilde{R}^{(1)}_{i} = 1, R_{i1} = 1) \\
        &=\sum_{r= 0,1}\E[Y_{i2} - Y_{i1} \mid D_{i} = 1,  \widetilde{R}^{(1)}_{i} = 1, R_{i1} = 1, R_{i2} = r] \Pr(R_{i2} = r \mid D_{i} = 1,  \widetilde{R}^{(1)}_{i} = 1, R_{i1} = 1) \\
        &\qquad + 
        \E[Y_{i2} - Y_{i1} \mid D_{i} = 1,  \widetilde{R}^{(1)}_{i} = 1, R_{i1} = 1, R_{i2} = 1]
        \Pr(R_{i2} = 0 \mid D_{i} = 1,  \widetilde{R}^{(1)}_{i} = 1, R_{i1} = 1)\\
        &\qquad -
        \E[Y_{i2} - Y_{i1} \mid D_{i} = 1,  \widetilde{R}^{(1)}_{i} = 1, R_{i1} = 1, R_{i2} = 1]
        \Pr(R_{i2} = 0 \mid D_{i} = 1,  \widetilde{R}^{(1)}_{i} = 1, R_{i1} = 1)\\
        &=\E[Y_{i2} - Y_{i1} \mid D_{i} = d,  \widetilde{R}^{(1)}_{i} = 1, R_{i1} = 1, R_{i2} = 1]\\ 
        &\qquad+\{\E[Y_{i2} - Y_{i1} \mid D_{i} = 1,  \widetilde{R}^{(1)}_{i} = 1, R_{i1} = 1, R_{i2} = 0] - \E[Y_{i2} - Y_{i1} \mid D_{i} = 1,  \widetilde{R}^{(1)}_{i} = 1, R_{i1} = 1, R_{i2} = 1]\}\\
        &\qquad\quad\times\Pr(R_{i2} = 0 \mid D_{i} = 1,  \widetilde{R}^{(1)}_{i} = 1, R_{i1} = 1)
    \end{align*}
    With similar steps, 
    \begin{align*}
        &\E[Y_{i2} - Y_{i1} \mid D_{i} = 1,  \widetilde{R}^{(1)}_{i} = 0, R_{i1} = 1] \\
        &=\E[Y_{i2} - Y_{i1} \mid D_{i} = 1,  \widetilde{R}^{(1)}_{i} = 0, R_{i1} = 1, R_{i2} = 1]\\ 
        &\qquad+\{\E[Y_{i2} - Y_{i1} \mid D_{i} = 1,  \widetilde{R}^{(1)}_{i} = 0, R_{i1} = 1, R_{i2} = 0] - \E[Y_{i2} - Y_{i1} \mid D_{i} = 1,  \widetilde{R}^{(1)}_{i} = 0, R_{i1} = 1, R_{i2} = 1]\}\\
        &\qquad\quad\times\Pr(R_{i2} = 0 \mid D_{i} = 1,  \widetilde{R}^{(1)}_{i} = 0, R_{i1} = 1) \\
        &=\E[Y_{i2} - Y_{i1} \mid D_{i} = 1,  \widetilde{R}^{(1)}_{i} = 0, R_{i1} = 1, R_{i2} = 1]\\ 
        &\qquad+\{\E[Y_{i2} - Y_{i1} \mid D_{i} = 1,  \widetilde{R}^{(1)}_{i} = 1, R_{i1} = 1, R_{i2} = 0] - \E[Y_{i2} - Y_{i1} \mid D_{i} = 1,  \widetilde{R}^{(1)}_{i} = 1, R_{i1} = 1, R_{i2} = 1]\}\\
        &\qquad\quad\times\Pr(R_{i2} = 0 \mid D_{i} = 1,  \widetilde{R}^{(1)}_{i} = 0, R_{i1} = 1)
    \end{align*}
    where the last equality holds by \textbf{bias homogeneity} in Assumption~\ref{a:missing_multiple_iv}.
    Because of the \textbf{parallel difference in trends}, we have
    \begin{align*}
        &\E[Y_{i2} - Y_{i1} \mid D_{i} = 1,  \widetilde{R}^{(1)}_{i} = 1, R_{i1} = 1, R_{i2} = 1]\\ 
        &\qquad+\{\E[Y_{i2} - Y_{i1} \mid D_{i} = 1,  \widetilde{R}^{(1)}_{i} = 1, R_{i1} = 1, R_{i2} = 0] - \E[Y_{i2} - Y_{i1} \mid D_{i} = 1,  \widetilde{R}^{(1)}_{i} = 1, R_{i1} = 1, R_{i2} = 1]\}\\
        &\qquad\quad\times\Pr(R_{i2} = 0 \mid D_{i} = 1,  \widetilde{R}^{(1)}_{i} = 1, R_{i1} = 1)\\
        &-\E[Y_{i2} - Y_{i1} \mid D_{i} = 1,  \widetilde{R}^{(1)}_{i} = 0, R_{i1} = 1, R_{i2} = 1]\\ 
        &\qquad-\{\E[Y_{i2} - Y_{i1} \mid D_{i} = 1,  \widetilde{R}^{(1)}_{i} = 1, R_{i1} = 1, R_{i2} = 0] - \E[Y_{i2} - Y_{i1} \mid D_{i} = 1,  \widetilde{R}^{(1)}_{i} = 1, R_{i1} = 1, R_{i2} = 1]\}\\
        &\qquad\quad\times\Pr(R_{i2} = 0 \mid D_{i} = 1,  \widetilde{R}^{(1)}_{i} = 0, R_{i1} = 1) \\
        &=\E[Y_{i2} - Y_{i1} \mid D_{i} = 1,  \widetilde{R}^{(2)}_{i} = 1, R_{i1} = 1, R_{i2} = 1]\\ 
        &\qquad+\{\E[Y_{i2} - Y_{i1} \mid D_{i} = 1,  \widetilde{R}^{(2)}_{i} = 1, R_{i1} = 1, R_{i2} = 0] - \E[Y_{i2} - Y_{i1} \mid D_{i} = 1,  \widetilde{R}^{(2)}_{i} = 1, R_{i1} = 1, R_{i2} = 1]\}\\
        &\qquad\quad\times\Pr(R_{i2} = 0 \mid D_{i} = 1,  \widetilde{R}^{(2)}_{i} = 1, R_{i1} = 1)\\
        &-\E[Y_{i2} - Y_{i1} \mid D_{i} = 1,  \widetilde{R}^{(2)}_{i} = 0, R_{i1} = 1, R_{i2} = 1]\\ 
        &\qquad-\{\E[Y_{i2} - Y_{i1} \mid D_{i} = 1,  \widetilde{R}^{(2)}_{i} = 1, R_{i1} = 1, R_{i2} = 0] - \E[Y_{i2} - Y_{i1} \mid D_{i} = 1,  \widetilde{R}^{(2)}_{i} = 1, R_{i1} = 1, R_{i2} = 1]\}\\
        &\qquad\quad\times\Pr(R_{i2} = 0 \mid D_{i} = 1,  \widetilde{R}^{(2)}_{i} = 0, R_{i1} = 1)
    \end{align*}
    By rearranging the terms,
    \begin{align*}
        &\E[Y_{i2} - Y_{i1} \mid D_{i} = 1,  \widetilde{R}^{(1)}_{i} = 1, R_{i1} = 1, R_{i2} = 1]-\E[Y_{i2} - Y_{i1} \mid D_{i} = 1,  \widetilde{R}^{(1)}_{i} = 0, R_{i1} = 1, R_{i2} = 1]\\ 
        &\qquad+\{\E[Y_{i2} - Y_{i1} \mid D_{i} = 1,  \widetilde{R}^{(1)}_{i} = 1, R_{i1} = 1, R_{i2} = 0] - \E[Y_{i2} - Y_{i1} \mid D_{i} = 1,  \widetilde{R}^{(1)}_{i} = 1, R_{i1} = 1, R_{i2} = 1]\}\\
        &\qquad\quad\times\{\Pr(R_{i2} = 0 \mid D_{i} = 1,  \widetilde{R}^{(1)}_{i} = 1, R_{i1} = 1) - \Pr(R_{i2} = 0 \mid D_{i} = 1,  \widetilde{R}^{(1)}_{i} = 0, R_{i1} = 1)\}\\
        &=\E[Y_{i2} - Y_{i1} \mid D_{i} = 1,  \widetilde{R}^{(2)}_{i} = 1, R_{i1} = 1, R_{i2} = 1]-\E[Y_{i2} - Y_{i1} \mid D_{i} = 1,  \widetilde{R}^{(2)}_{i} = 0, R_{i1} = 1, R_{i2} = 1]\\ 
        &\qquad+\{\E[Y_{i2} - Y_{i1} \mid D_{i} = 1,  \widetilde{R}^{(2)}_{i} = 1, R_{i1} = 1, R_{i2} = 0] - \E[Y_{i2} - Y_{i1} \mid D_{i} = 1,  \widetilde{R}^{(2)}_{i} = 1, R_{i1} = 1, R_{i2} = 1]\}\\
        &\qquad\quad\times\{\Pr(R_{i2} = 0 \mid D_{i} = 1,  \widetilde{R}^{(2)}_{i} = 1, R_{i1} = 1) - \Pr(R_{i2} = 0 \mid D_{i} = 1,  \widetilde{R}^{(2)}_{i} = 0, R_{i1} = 1)\}
    \end{align*}
    Again, by \textbf{bias homogeneity}, we have
    \begin{align*}
        &\E[Y_{i2} - Y_{i1} \mid D_{i} = 1,  \widetilde{R}^{(1)}_{i} = 1, R_{i1} = 1, R_{i2} = 1]-\E[Y_{i2} - Y_{i1} \mid D_{i} = 1,  \widetilde{R}^{(1)}_{i} = 0, R_{i1} = 1, R_{i2} = 1]\\
        &- \{\E[Y_{i2} - Y_{i1} \mid D_{i} = 1,  \widetilde{R}^{(2)}_{i} = 1, R_{i1} = 1, R_{i2} = 1]-\E[Y_{i2} - Y_{i1} \mid D_{i} = 1,  \widetilde{R}^{(2)}_{i} = 0, R_{i1} = 1, R_{i2} = 1]\}\\
        &=\{\E[Y_{i2} - Y_{i1} \mid D_{i} = 1,  \widetilde{R}^{(2)}_{i} = 1, R_{i1} = 1, R_{i2} = 0] - \E[Y_{i2} - Y_{i1} \mid D_{i} = 1,  \widetilde{R}^{(2)}_{i} = 1, R_{i1} = 1, R_{i2} = 1]\}\\
        &\qquad\quad\times\big[\{\Pr(R_{i2} = 0 \mid D_{i} = 1,  \widetilde{R}^{(2)}_{i} = 1, R_{i1} = 1) - \Pr(R_{i2} = 0 \mid D_{i} = 1,  \widetilde{R}^{(2)}_{i} = 0, R_{i1} = 1)\} \\
        &\qquad\qquad- \{\Pr(R_{i2} = 0 \mid D_{i} = 1,  \widetilde{R}^{(1)}_{i} = 1, R_{i1} = 1) - \Pr(R_{i2} = 0 \mid D_{i} = 1,  \widetilde{R}^{(1)}_{i} = 0, R_{i1} = 1)\} \big]
    \end{align*}
    We can rearrange the terms to get the desired result.
    Similarly, we can show the equation with $d = 0$.
\end{proof}

\newpage
\section{Identification of ATT with Principal Ignorability} \label{app:pi}

\begin{namedassumption}{a:principal-ignorability}{Principal ignorability conditioning on the treatment}{
    \begin{equation*}
        (Y_{i2}-Y_{i1}) \indep S_{i} \mid X_{i}, D_{i}
    \end{equation*}
    where $X_i$ is the pre-treatment covariates.}
\end{namedassumption}

\namedCref{a:principal-ignorability} is a variant of principal ignorability assumption \citep{ding2016}, that has been adjusted to parallel trends setup. It implies that
\begin{align*}
    \E[Y_{i2}(1)-Y_{i1} \mid S_{i} = s,  X_{i}, D_{i} = 1] &= \E[Y_{i2}(1)-Y_{i1} \mid S_{i} = s^\prime,  X_{i}, D_{i} = 1] \\
    \E[Y_{i2}(0)-Y_{i1} \mid S_{i} = s,  X_{i}, D_{i} = 0] &= \E[Y_{i2}(0)-Y_{i1} \mid S_{i} = s^\prime,  X_{i}, D_{i} = 0]
\end{align*}
To further illustrate this, consider the following set of assumptions that are stronger yet intuitive --- if these two assumptions hold, then \namedCref{a:principal-ignorability} also holds:
\begin{equation*}
    Y_{i2}(1)-Y_{i2}(0) \indep S_{i} \mid X_{i}, D_{i}\text{ and }Y_{i2}(0)-Y_{i1} \indep S_{i} \mid X_{i}, D_{i}.
\end{equation*}
This implies (1) homogeneous effect across principal strata conditioning on covariates and treatment and (2) conditional parallel trends across principal strata within each treatment group.
Another set of assumptions that are more minimal are:
\begin{equation*}
    Y_{i2}(1) \indep S_{i} \mid X_{i}, D_{i}\text{ and }Y_{i2}(0) \indep S_{i} \mid X_{i}, D_{i}.
\end{equation*}
This is equivalent to principal ignorability assumption from the original paper (which assumes a randomized treatment), yet conditioning on the treatment as well. 

\begin{namedremark}{r:pi}{Principal ignorability under PT as MAR}{
How does \namedCref{a:principal-ignorability} related to missing at random assumption? Specifically, we can compare \Cref{a:principal-ignorability} with the following missing at random assumption:
    \begin{equation*}
        (Y_{i2}-Y_{i1}) \indep R_{i} \mid X_{i}, D_{i}.
    \end{equation*}
Note that $S_i \equiv (R_i(1), R_i(0))$. Thus \namedCref{a:principal-ignorability} which assumes a joint independence implies the marginal independence above. As a result, one may think of this assumtion as a specific type of outcome missing at random.}
\end{namedremark}

Here, we introduce two parallel trend assumptions conditioning on pre-treatment covaraites.

\begin{namedassumption}{a:pt-principal-cov}{Principal strata parallel trend with covariates}{
\begin{align*}
        \E[Y_{i2}(0) - Y_{i1} \mid D_{i} = 0, S_{i} = s, X_{i} = x] = \E[Y_{i2}(0) - Y_{i1} \mid D_{i} = 1, S_{i} = s, X_{i} = x]
    \end{align*}
    for $s \in \{(0,0), (0,1), (1,0), (1,1)\}$.
}
\end{namedassumption}

\begin{namedassumption}{a:pt-missing-cov}{Parallel trend of missingness with covariates}{
    \begin{equation*}
        \E[R_{i2}(0) - R_{i1}(0) \mid D_{i} = 1, X_{i} = x] = \E[R_{i2}(0) - R_{i1}(0) \mid D_{i} = 0, X_{i} = x]
    \end{equation*}}
\end{namedassumption}

\begin{namedproposition}{prop:principal-score}{Principal score}{
    Let $e_{r_{1}, r_{0}}(x) = \Pr(R_{i2}(1) = r_{1}, R_{i2}(0) = r_{0} \mid D_{i} = 1, X_{i} = x)$ for $r_{1}, r_{0} = 0,1$. Then, under \namedCref{a:monotonicity} and \namedCref{a:pt-missing-cov}, we have
    \begin{align*}
        e_{11}(x) &= \Pr(R_{i1} = 1 \mid D_{i} = 1, X_{i} = x) \\
        &\quad + \Pr(R_{i2} = 1 \mid D_{i} = 0, X_{i} = x) - \Pr(R_{i1} = 1 \mid D_{i} = 0, X_{i} = x) \\
        e_{00}(x) &= \Pr(R_{i2} = 0 \mid D_{i} = 1, X_{i} = x) \\
        e_{10}(x) &= \Pr(R_{i2} = 1 \mid D_{i} = 1, X_{i} = x) - e_{11}(x) 
    \end{align*}}
\end{namedproposition}

\begin{namedtheorem}{thm:pi-ident}{Identification with principal ignorability}{
    Under \namedCref{a:monotonicity}, \namedCref{a:principal-ignorability}, \namedCref{a:pt-principal-cov}, and \namedCref{a:pt-missing-cov}, we have
        \begin{align*}
            &\E[Y_{i2}(1) - Y_{i2}(0) \mid D_{i} = 1, R_{i2}(1) = r_{1}, R_{i2}(0) = r_{0}] \\
            &\quad = \E[w_{r_{1}, r_{0}}(X_{i}) Y_{i2} \mid D_{i} = 1, R_{i2} = 1]  - \E[w_{r_{1}, r_{0}}(X_{i}) Y_{i1} \mid D_{i} = 1] \\
            &\qquad- \E[w_{r_{1}, r_{0}}(X_{i}) Y_{i2} \mid D_{i} = 0, R_{i2} = 1] + \E[w_{r_{1}, r_{0}}(X_{i}) Y_{i1} \mid D_{i} = 0]
        \end{align*}
        where 
        \begin{equation*}
            w_{r_{1}, r_{0}}(x) = \frac{e_{r_{1}, r_{0}}(x)}{\E[e_{r_{1}, r_{0}}(X_{i})]}
        \end{equation*}
        for $(r_{1}, r_{0}) \in \{(1,1), (0,0), (1,0)\}$.
        Accordingly, we have
        \begin{align*}
            \text{ATT} &= \sum_{(r_{1}, r_{0})}\E[e_{r_{1}, r_{0}}(X_{i}) Y_{i2} \mid D_{i} = 1, R_{i2} = 1]  - \E[e_{r_{1}, r_{0}}(X_{i}) Y_{i1} \mid D_{i} = 1] \\
            &\qquad- \E[e_{r_{1}, r_{0}}(X_{i}) Y_{i2} \mid D_{i} = 0, R_{i2} = 1] + \E[e_{r_{1}, r_{0}}(X_{i}) Y_{i1} \mid D_{i} = 0]
        \end{align*}}
\end{namedtheorem}

\paragraph{Limitations of Principal Ignorability}

A list of notable limitations of principal ignorability under this setup is:
\begin{itemize}
    \item It is in essence a missing at random assumption.
    \item In practice, it is likely that we also have missingness in pre-treatment covariates.
    \item With highdimensional pre-treatment covariates, one may potentially need to assume a parametric model for the estimation.
\end{itemize}

\end{document}